\documentclass[english]{article}
\usepackage[T1]{fontenc}
\usepackage[latin9]{inputenc}
\usepackage[a4paper]{geometry}
\usepackage{color}
\usepackage{array}
\usepackage{float}
\usepackage{multirow}
\usepackage{amsmath}
\usepackage{amsthm}
\usepackage{amssymb}
\usepackage{stmaryrd}
\usepackage{graphicx}
\usepackage{setspace}
\usepackage{esint}
\usepackage[authoryear]{natbib}
\PassOptionsToPackage{normalem}{ulem}
\usepackage{ulem}
\makeatletter

\providecommand{\tabularnewline}{\\}
\floatstyle{ruled}
\newfloat{algorithm}{tbp}{loa}
\providecommand{\algorithmname}{Algorithm}
\floatname{algorithm}{\protect\algorithmname}

\theoremstyle{plain}
\newtheorem{thm}{\protect\theoremname}
  \theoremstyle{remark}
  \newtheorem{rem}{\protect\remarkname}
  \theoremstyle{plain}
  \newtheorem{lem}{\protect\lemmaname}
  \theoremstyle{plain}
  \newtheorem{prop}{\protect\propositionname}
  \theoremstyle{plain}
  \newtheorem{cor}{\protect\corollaryname}
 \theoremstyle{definition}
  \newtheorem{example}{\protect\examplename}

\newcounter{hypA}
\newenvironment{condition}{\refstepcounter{hypA}\begin{itemize}
\item[({\bf H})]}{\end{itemize}}

\newcommand{\iid}{\stackrel{\mathrm{iid}}{\sim}}

\makeatother

\usepackage{babel}
  \providecommand{\examplename}{Example}
  \providecommand{\lemmaname}{Lemma}
  \providecommand{\propositionname}{Proposition}
  \providecommand{\remarkname}{Remark}
\providecommand{\corollaryname}{Corollary}
\providecommand{\theoremname}{Theorem}

\begin{document}

\title{Calculating principal eigen-functions of non-negative integral kernels:
particle approximations and applications}

\author{Nick Whiteley\thanks{School of Mathematics, University of Bristol, University Walk, Bristol,
BS8 1TW.} \ and\  Nikolas Kantas\thanks{Department of Mathematics, Imperial College London, South Kensington
Campus, London SW7 2AZ. }}
\maketitle
\begin{abstract}
Often in applications such as rare events estimation or optimal control
it is required that one calculates the principal eigen-function and
eigen-value of a non-negative integral kernel. Except in the finite-dimensional
case, usually neither the principal eigen-function nor the eigen-value
can be computed exactly. In this paper, we develop numerical approximations
for these quantities. We show how a generic interacting particle algorithm
can be used to deliver numerical approximations of the eigen-quantities
and the associated so-called ``twisted'' Markov kernel as well as
how these approximations are relevant to the aforementioned applications.
In addition, we study a collection of random integral operators underlying
the algorithm, address some of their mean and path-wise properties,
and obtain $L_{r}$ error estimates\textcolor{black}{. Finally, numerical
examples are provided in the context of importance sampling for computing
tail probabilities of Markov chains and computing value functions
for a class of stochastic optimal control problems.\medskip{}
}\\
\textcolor{black}{Keywords: interacting particle methods, eigen-functions,
rare events estimation, optimal control, diffusion Monte Carlo}
\end{abstract}

\section{Introduction}

On a state space $\mathsf{X}$ consider a bounded function $G:\mathsf{X}\rightarrow\mathbb{R}_{+}$,
a Markov probability kernel $M$. The central object of interest in
this paper is the integral kernel $Q$ given by
\[
Q(x,dx^{\prime}):=G(x)M(x,dx^{\prime}).
\]
Under some regularity assumptions, $Q$ has an isolated, real, maximal
eigen-value $\lambda_{\star}$, with which is associated a positive
(right) eigen-function $h_{\star}$, 
\begin{equation}
Q(h_{\star})=\lambda_{\star}h_{\star},\label{eq:sec1_eig_func}
\end{equation}
where for a function $\varphi$ on $\mathsf{X}$, we write $Q(\varphi)(x):=\int Q(x,dx^{\prime})\varphi(x^{\prime})$.
When $\mathsf{X}$ is finite set, $\lambda_{\star}$ is the Perron-Frobenius
eigen-value and $h_{\star}$ the right eigen-vector. In this paper
we are interested in the case where $\mathsf{X}$ is a general space,
so not necessarily finite or countable. In general state spaces an
extended Perron-Frobenius theory applies, (see \citet{mc:the:N04}
for an account), but in most cases $\lambda_{\star},h_{\star}$ cannot
be determined analytically, so numerical approximations are required
and this is what this paper aims to address.

Treatment of the existence of $\lambda_{\star}$ and $h_{\star}$
outside of settings in which $\mathsf{X}$ is a finite set dates at
least as far as \citep{kolmogorov1938losung,yaglom1947certain,Harris63},
where $Q$ arose as a conditional moment measure associated with a
branching process; see \citet{collet2012quasi} for a modern perspective
in the context of quasi-stationary distributions and stochastic processes
conditioned on long-term survival. In addition, $Q$ and $h_{\star}$
have often appeared as critical quantities in various more recent
applications. In statistical mechanics $Q$ corresponds to the Hamiltonian
and $h_{\star}$ could be viewed as the Schrödinger ground energy
state for molecules, e.g. \citep{rousset2006control,makrini2007diffusion}.
Similarly, in particle physics $Q\bigl(1\bigr)\left(x\right)$ can
be used to model the one-step probability of survival of of a particle
moving in an absorbing medium \citep[Chapter 7]{del2013mean}, \citep{smc:the:DMD04}.
In stochastic optimal control, $Q$ arises naturally as a multiplicative
Bellman or Dynamic Programming operator in discrete time problems
when a Kullback-Leibler divergence term is used in the stage cost
\citep{albertini1988logarithmic,todorov2008general,dvijotham2011unified}
or in particular continuous time models with affine dynamics in the
control and additive costs that are quadratic to the control input;
see \citep{fleming1982logarithmic,sheu1984stochastic} or \citep{todorov2008general,theodorou2010generalized,kappen2005linear}
for more details. In these specific control problems, $h_{\star}$
can be viewed as a logarithmic transformation of the value function.
Finally, $h_{\star}$ appears in the large deviations theory of Markov
chains, see for example \citep{ney1987markov}; if $(X_{n};n\geq0)$
is a Markov chain with transition kernel $M$, initialized from $X_{0}=x$,
$U$ an appropriate function and $G(x):=e^{\alpha U(x)}$ for a particular
value of $\alpha,$ then it is only and explicitly through $h_{\star}(x)$
that the initial condition enters Bahadur-Rao-type asymptotics associated
with partial sums $\sum_{p=0}^{n-1}U(X_{p})$ \citep{mc:theory:KM03}. 

A related object of interest in many applications of interest is the
``twisted'' Markov kernel:
\begin{equation}
P_{\star}(x,dx^{\prime}):=\frac{Q(x,dx^{\prime})h_{\star}(x^{\prime})}{h_{\star}(x)\lambda_{\star}},\label{eq:twisted_intro}
\end{equation}
which is also known as h-process kernel \citet{collet2012quasi} or
Doob's h-transform \citet[Section III.29]{rogers1996diffusions}.
Particular instances of $P_{\star}$ define optimal changes of measure
in methods for estimating rare event probabilities, such as \textcolor{black}{for
tail probabilities of Markov chains} \citep{bucklew1990monte,dupuis2005dynamic}.
In the discrete time control problems mentioned above $P_{\star}$
defines the optimally controlled Markov transition kernel. In the
context of particle motion in absorbing media $P_{\star}$ is the
Markov transition kernel of a particle conditional on long-term survival
\citet[Section 7.2 pages 223-226]{del2013mean}, and for, multi-type
branching processes, $P_{\star}$ defines a transformation from supercritical
to critical \citep{athreya2000change}. 

Of course the eigen-function equation (\ref{eq:sec1_eig_func}) is
just one side of the story. Accompanying $h_{\star}$ is a (left)
eigen-measure, which under certain conditions can be normalized to
a probability measure $\eta_{\star}$, 
\begin{equation}
\eta_{\star}Q=\lambda_{\star}\eta_{\star},\label{eq:sec1_eig_meas}
\end{equation}
where for a measure $\eta$, we write $\eta Q(\cdot):=\int\eta\left(dx\right)Q(x,\cdot)$.
\citet{smc:the:dMM03} studied the non-linear operator on measures
\begin{equation}
\Phi:\eta\mapsto\frac{\eta Q}{\eta Q\left(1\right)},\label{eq:Phi_defn}
\end{equation}
(where $1$ is the unit function on $\mathsf{X}$). Under regularity
assumptions, for sufficiently large $n$, the $n$-fold iterated operator
$\Phi^{(n)}$ is contractive with respect to total-variation norm
and $\eta_{\star}$ is its unique fixed point. Indeed integrating
both sides of (\ref{eq:sec1_eig_meas}) yields $\eta_{\star}Q\left(1\right)=\lambda_{\star}$
so that $\Phi(\eta_{\star})=\eta_{\star}$ is a re-writing of (\ref{eq:sec1_eig_meas});
see \citet{smc:the:dMM03,smc:the:DMD04} for more details. In these
papers the authors suggested and analyzed an interacting particle
algorithm whose evolution is defined through $\Phi$ and which can
be used to approximate $\eta_{\star}$ and $\lambda_{\star}$. When
$M$ is reversible, $h_{\star}$ provides a density of $\eta_{\star}$.
In this case the particle algorithm analyzed in \citet{smc:the:dMM03}
and \citet{smc:the:DMD04} has also appeared in the statistical mechanics
literature, \citet{assaraf2000diffusion,rousset2006control,makrini2007diffusion},
under the name Diffusion Monte Carlo and has been used to provide
estimates of $h_{\star}$ and $\lambda_{\star}$. Finally, we mention
the Flemming-Viot particle system in \citet{burdzy2000fleming}, where
the authors without using any reversibility assumptions use the continuous
time analog of \citet{smc:the:dMM03,smc:the:DMD04} to perform spectral
analysis of the Laplacian with Dirichlet boundary conditions.

The contributions of the paper are summarized as follows:
\begin{itemize}
\item We propose an interacting particle algorithm for approximating $h_{\star}$
and $P_{\star}$ numerically. Our algorithm does not hinge upon reversibility
assumptions on $M$ and is similar in structure to one proposed by
\citet{del2011robustness,del2012snell} for the rather different purpose
of numerically solving optimal stopping problems. The novelty of our
approach is that we obtain a particle approximation of $P_{\star}$
that is easy to sample from\textcolor{black}{, which is an important
factor in applications.}
\item We apply our method to two problems. The first application is a Markov
chain rare-event problem, here our method allows us to unbiasedly
estimate tail probabilities for additive functions of Markov chains
by importance sampling and $P_{\star}$ defines an optimal change
of measure derived by \citet{bucklew1990monte}, which we are able
to approximate. The second application is an optimal control problem
as studied in \citep{albertini1988logarithmic,todorov2008general,dvijotham2011unified},
in which the cost function involves a Kullback-Leibler divergence
term. Here $P_{\star}$ specifies the optimal dynamics for a controlled
Markov chain.
\item We study the convergence properties of our algorithm, in particular
deriving moment bounds for the errors in approximation of $h_{\star}$
and $P_{\star}$, and we derive certain path-wise stability properties
of random operators obtained from our algorithm, demonstrating that
they inherit the ``tendency to rank-one'' behavior of the iterated
operator $\lambda_{\star}^{-1}Q^{(n)}$. 
\end{itemize}

\subsection{Organization of the paper}

The remainder of this paper is structured as follows. Section \ref{sec:The-eigen-problem}
provides notation and sets out the eigen-problem. Section \ref{sec:Applications}
presents the motivating applications. In Section \ref{sec:Overview}
we present the particle algorithm and state the our results regarding
various properties of the particle approximations. More details and
precise statements for these are found in Section \ref{sec:Particle-approximations}.
Section \ref{sec:Examples} contains numerical results for the application.
Some concluding remarks and possible extensions are presented in Section
\ref{sec:Extensions}. Finally, various proofs are contained in the
appendix.

\section{The eigen-problem\label{sec:The-eigen-problem}}

\subsection{Notation and assumptions}

Let $\mathsf{X}$ be a state space endowed with a countably generated
$\sigma$-algebra $\mathcal{B}$ and let $\mathcal{L}$ be the Banach
space of real-valued, $\mathcal{B}$-measurable, bounded functions
on $\mathsf{X}$ endowed with the infinity norm $\left\Vert f\right\Vert :=\sup_{x\in\mathsf{X}}\left|f(x)\right|$.
For a possibly signed measure $\eta$, a function $\varphi$, and
a possibly signed integral kernel $K$ we write $\mu(\varphi):=\int\varphi(x)\mu(dx)$,
$K(\varphi)(x):=\int K(x,dy)\varphi(dy)$, and $\mu K(\cdot):=\int\mu(dx)K(x,\cdot)$,
and the rank-one kernel $(\varphi\otimes\eta)(x,dx^{\prime}):=\varphi(x)\eta(dx^{\prime})$. 

The collection of probability measures on $\left(\mathsf{X},\mathcal{B}\right)$
is denoted by $\mathcal{P}$ and the total variation norm for possibly
signed measures is denoted $\left\Vert \eta\right\Vert :=\sup_{\varphi:\left|\varphi\right|\leq1}\left|\eta(\varphi)\right|$.
The operator norm corresponding to $\mathcal{L}$ is 
\[
\interleave K\interleave:=\sup_{\varphi:\left|\varphi\right|\leq1}\left\Vert K(\varphi)\right\Vert .
\]
The $n$-fold iterate of $K$ is denoted by $K^{\left(n\right)}$
and for $\left(K_{n};n\geq1\right)$ a collection of integral kernels
and any $0\leq p\leq n$, we write
\begin{equation}
K_{p,n}:=Id,\;\;p=n,\quad\quad K_{p,n}:=K_{p+1}\cdots K_{n},\quad n>p.\label{eq:semi_g_notation}
\end{equation}

Throughout the paper, we denote by $G:\mathsf{X}\rightarrow(0,\infty)$
is a $\mathcal{B}$-measurable, bounded function and let $M:\mathsf{X}\times\mathcal{B}\rightarrow[0,1]$
be a Markov kernel, then define the integral kernel $Q(x,dy):=G(x)M(x,dy)$.
We have 
\[
\interleave Q\interleave=\sup_{x\in\mathsf{X}}Q(1)(x)=\sup_{x\in\mathsf{X}}G(x),
\]
and $\interleave Q\interleave<\infty$ due to $G$ being bounded.
The spectral radius of $Q$ as a bounded linear operator on $\mathcal{L}$
is 
\[
\xi:=\lim_{n\rightarrow\infty}\interleave Q^{\left(n\right)}\interleave^{1/n}
\]
where the limit always exists, since the operator norm is sub-multiplicative. 

For two probability measures $\mu,\nu\in\mathcal{P}$ we will denote
the Kullback-Leibler divergence or relative entropy as 
\[
{\cal K}{\cal L}(\left.\mu\right\Vert \nu):=\begin{cases}
\int\log\left(\frac{\mathrm{d}\mu}{\mathrm{d}\nu}\right)d\mu & \mbox{ if \ensuremath{\mu\ll\nu}},\\
+\infty & \mbox{ otherwise}.
\end{cases}
\]
For any sequence $\left(a_{n};n\geq1\right)$ and $\ell>p,$ we take
$\prod_{n=\ell}^{p}a_{n}=1$ by convention. The unit function on $\mathsf{X}$
or Cartesian products thereof is denoted by $1$. We will write the
indicator function $\mathbb{I}[\cdot]$ or sometimes $\mathbb{I}_{A}$
for a set $A\subset\mathsf{X}$. Unless stated otherwise, we will
assume throughout: 

\begin{condition}\label{hyp:mixing}

there exists a probability measure $\nu$ such that for all $x$,
$Q(x,\cdot)$ is equivalent to $\nu.$ There exist constants $0<\epsilon^{-},\epsilon^{+}<\infty$
such that the corresponding Radon-Nikodym derivative, denoted by $q(x,x^{\prime}):=\dfrac{dQ(x,\cdot)}{d\nu}(x^{\prime})$
satisfies
\[
\epsilon^{-}\leq q(x,x^{\prime})\leq\epsilon^{+},\quad\forall x,x^{\prime}\in\mathsf{X}.
\]

\end{condition}In some places it will be convenient to use the implication
of \textbf{(H)}
\[
\epsilon^{-}\nu(\cdot)\leq Q(x,\cdot)\leq\epsilon^{+}\nu(\cdot),\quad\forall x\in\mathsf{X}.
\]
The uniform recurrence of $Q$ in Assumption \textbf{(H)} is a quite
strong assumption, but has been used extensively in both the particle
filtering literature (\citet{del2013mean,smc:theory:Dm04,douc2010sequential})
and the rare events literature related to tail probabilities of interest
here \textcolor{black}{(\citep{bucklew1990monte,dupuis2005dynamic,chan:lai}}).
It rules out kernels of the form $Q(x,\cdot)=p\delta_{x}(\cdot)+\ldots$,
and rarely holds when $\mathsf{X}$ is non-compact, but allows a relatively
straightforward treatment of the eigen-problem and the particle algorithm.
The eigen-quantities of interest exist under much weaker assumptions,
and a result similar to Theorem \ref{thm:MET} presented later in
Section (\ref{sec:Multiplicative-Ergodicity}) can be obtained for
non-compact $\mathsf{X}$ in a weighted $\infty$-norm setting under
quite flexible Lyapunov drift conditions \citep{mc:theory:KM03,whiteley2012}.
The details, however, would necessitate a much more complicated presentation,
and obtaining error bounds of the sort we do for the particle approximations,
under assumptions much weaker than \textbf{(H)} seems very challenging.

\subsection{Existence and other properties of eigen-quantities\label{sec:Multiplicative-Ergodicity}}

From the minorization part of \textbf{(H)} 
\[
\nu Q^{(n+m-1)}(1)\epsilon^{-}=\nu Q^{(n)}Q^{(m-1)}(1)\epsilon^{-}\geq\nu Q^{(n-1)}(1)\epsilon^{-}\nu Q^{(m-1)}(1)\epsilon^{-},
\]
so by Fekete's lemma, the following limit exists,
\begin{equation}
\Lambda_{\star}:=\lim_{n\rightarrow\infty}\frac{1}{n}\log\nu Q^{\left(n-1\right)}(1)\epsilon^{-}=\sup_{n\geq1}\frac{1}{n}\log\nu Q^{\left(n-1\right)}(1)\epsilon^{-},\label{eq:Lambda}
\end{equation}
Define
\begin{equation}
\lambda_{\star}:=\exp(\Lambda_{\star}),\label{eq:lambda}
\end{equation}

The proof of Theorem \ref{thm:MET} is given in the Appendix, and
it involves gathering together various arguments from \citet{mc:the:N04},
which we recount there for the reader's convenience.
\begin{thm}
\label{thm:MET} The spectral radius of $Q$, $\lim_{n\rightarrow\infty}\interleave Q^{\left(n\right)}\interleave^{1/n}$,
coincides with $\lambda_{\star}$. There exists a unique probability
measure $\eta_{\star}$ and $\nu$-essentially unique positive function
$h_{\star}$ satisfying
\begin{equation}
\eta_{\star}Q=\lambda_{\star}\eta_{\star},\quad\quad Q(h_{\star})=\lambda_{\star}h_{\star},\quad\quad\eta_{\star}(h_{\star})=1.\label{eq:eigen_star}
\end{equation}
Furthermore,
\begin{equation}
\frac{\epsilon^{-}}{\epsilon^{+}}\leq h_{\star}(x)\leq\frac{\epsilon^{+}}{\epsilon^{-}},\quad\forall x\in\mathsf{X},\label{eq:eta_in_P_h_bound}
\end{equation}
 $P_{\star}$ has a unique invariant probability distribution, denoted
by $\pi_{\star}$, such that $d\pi_{\star}/d\eta_{\star}=h_{\star}$
and for all $n\geq1$,
\begin{eqnarray}
\interleave P_{\star}^{(n)}-1\otimes\pi_{\star}\interleave & \leq & 2\rho^{n}\label{eq:P_star_ergo}\\
\interleave\lambda_{\star}^{-n}Q^{\left(n\right)}-h_{\star}\otimes\eta_{\star}\interleave & \leq & 2\rho^{n}\left(\frac{\epsilon^{+}}{\epsilon^{-}}\right)^{2},\label{eq:MET_bound}
\end{eqnarray}
where $\rho:=1-\left(\epsilon^{-}/\epsilon^{+}\right)$.\end{thm}
\begin{rem}
The bound in (\ref{eq:MET_bound}) can be understood as describing
``tendency to rank-one'' of the iterated kernel $\lambda_{\star}^{-n}Q^{\left(n\right)}$,
this kind of result is sometimes referred to as a Multiplicative Ergodic
Theorem (MET) \citep{mc:theory:KM03}.
\end{rem}

\subsection{Deterministic approximations\label{sec:Deterministic_approximations_intro} }

We proceed by defining the deterministic forward-backward recursions
which will be used to approximate $\eta_{\star}$, $\lambda_{\star}$,
$h_{\star}$ and $P_{\star}$. These will appear throughout the remainder
of the paper.

\subsubsection*{Forward recursion for measures $\eta_{n}$}

Define the probability measures $(\eta_{n};n\geq0)$ and numbers $(\lambda_{n};n\geq0)$
by
\begin{equation}
\eta_{0}:=\mu,\quad\quad\eta_{n}:=\frac{\mu Q^{\left(n\right)}}{\mu Q^{\left(n\right)}(1)},\;n\geq1,\quad\quad\lambda_{n}:=\eta_{n}(G),\;n\geq0.\label{eq:eta_defn}
\end{equation}
Immediately from (\ref{eq:eta_defn}) we have the product formula:
\begin{equation}
\eta_{p}Q^{\left(n-p\right)}(1)=\prod_{\ell=p}^{n-1}\frac{\eta_{p}Q^{(\ell-p+1)}(1)}{\eta_{p}Q^{(\ell-p)}(1)}=\prod_{\ell=p}^{n-1}\eta_{\ell}(G)=\prod_{\ell=p}^{n-1}\lambda_{\ell},\quad p\leq n,\label{eq:product formula}
\end{equation}
and we note that
\begin{equation}
\eta_{n}=\Phi(\eta_{n-1}),\quad n\geq1,\label{eq:eta_Phi_recurse}
\end{equation}
with $\Phi$ defined earlier in (\ref{eq:Phi_defn}). Straightforward
manipulations show that under \textbf{(H)}, for any $n\geq1$, $\eta_{n}$
is equivalent to $\nu$.

\subsubsection*{Backward recursion for functions $h_{p,n}$}

Define the sequence of non-negative functions $(h_{p,n};0\leq p\leq n)$
as follows:
\begin{equation}
h_{n,n}(x):=1,\quad\quad h_{p,n}(x):=\frac{Q^{\left(n-p\right)}(1)(x)}{\eta_{p}Q^{\left(n-p\right)}(1)},\quad\quad0\leq p<n,x\in\mathsf{X}.\label{eq:h_p_n_defn}
\end{equation}

\begin{rem}
It should be noted that $(\eta_{n})$, $(\lambda_{n})$ and $(h_{p,n},P_{\left(p,n\right)})$
depend implicitly on the initial measure $\mu$.
\end{rem}

\subsubsection*{Properties}

The following lemma shows that the quantities $(\eta_{n})$, $(h_{p,n})$,
$(\lambda_{n})$ satisfy recursive relationships similar to the eigen-measure/function/value
equations in (\ref{eq:eigen_star}).
\begin{lem}
\label{lem:approx_eigen}The probability measures $(\eta_{n})$, functions
$(h_{p,n})$ and numbers $(\lambda_{n})$ satisfy
\begin{equation}
\eta_{p}Q=\lambda_{p}\eta_{p+1},\quad\quad Q(h_{p+1,n})=\lambda_{p}h_{p,n},\quad\quad\eta_{p}(h_{p,n})=1,\quad0\leq p\leq n.\label{eq:h_recurse}
\end{equation}
\end{lem}
\begin{proof}
The measure equation is just a rearrangement of (\ref{eq:eta_Phi_recurse}).
The function equation is due to the definition of $(h_{p,n})$ and
the product formula (\ref{eq:product formula}), as
\[
h_{p,n}=\frac{Q^{\left(n-p\right)}(1)}{\eta_{p}Q^{\left(n-p\right)}(1)}=\frac{\eta_{p+1}Q^{\left(n-p-1\right)}(1)}{\eta_{p}Q^{\left(n-p\right)}(1)}Q(h_{p+1,n})=\frac{1}{\lambda_{p}}Q(h_{p+1,n}).
\]
The final equality in (\ref{eq:h_recurse}) holds due to the definition
(\ref{eq:h_p_n_defn}).
\end{proof}
Lets define now the Markov probability kernel
\begin{equation}
P_{(p,n)}(x,dx^{\prime}):=\frac{Q(x,dx^{\prime})h_{p,n}(x^{\prime})}{\lambda_{p-1}h_{p-1,n}(x)},\label{eq:P_p_n_defn}
\end{equation}
where Lemma \ref{lem:approx_eigen} ensures it is indeed Markov. We
proceed with a proposition that can be used to justify the choice
of $(\eta_{n})$, $(h_{p,n})$, $(P_{(p,n)})$ as intermediate approximations
of $\eta_{\star}$, $h_{\star}$, $P_{\star}$ respectively. The proof
is in the Appendix.
\begin{prop}
\label{prop:_h_n_h_star_bound}For any $0\leq p\leq n$, 
\begin{eqnarray}
\left\Vert \eta_{n}-\eta_{\star}\right\Vert  & \leq & \rho^{n}C_{\eta},\label{eq:eta_converge}\\
\left\Vert h_{p,n}-h_{\star}\right\Vert  & \leq & \rho^{\left(n-p\right)\wedge p}C_{h},\label{eq:h_converge}\\
\interleave P_{(p,n)}-P_{\star}\interleave & \leq & \rho^{\left(n-p\right)\wedge p}C_{P},\label{eq:P_converge}
\end{eqnarray}
with 
\begin{eqnarray*}
\rho & := & 1-\left(\epsilon^{-}/\epsilon^{+}\right)\\
C_{\eta} & := & 4\left(\epsilon^{+}/\epsilon^{-}\right)^{3}\\
C_{h} & := & 2\left(\epsilon^{+}/\epsilon^{-}\right)^{2}\left[1+\left(\epsilon^{+}/\epsilon^{-}\right)+2\left(\epsilon^{+}/\epsilon^{-}\right)^{3}\right]\\
C_{P} & := & 2C_{h}\left(\epsilon^{+}/\epsilon^{-}\right)^{2}+C_{\eta}\rho^{-1}\left(\epsilon^{+}/\epsilon^{-}\right)
\end{eqnarray*}
having no dependence on the initial measure $\mu$.\end{prop}
\begin{rem}
Exponential convergence of the general form (\ref{eq:eta_converge})
has already been established in, for example, \citet{smc:the:DMD04}
using Dobrushin arguments for a collection of inhomogeneous Markov
kernels, but the rate obtained there is $\tilde{\rho}:=1-\left(\epsilon^{-}/\epsilon^{+}\right)^{2}$
as opposed to $\rho$. The proof of Proposition \ref{prop:_h_n_h_star_bound}
uses the MET bound of equation (\ref{eq:MET_bound}) and, as may be
seen in the proof of Theorem \ref{thm:MET}, the rate $\rho$ is inherited
from the uniform geometric ergodicity of $P_{\star}$ as per (\ref{eq:P_star_ergo}).
This is the source of the improved rate.
\end{rem}

\section{Applications\label{sec:Applications}}

We will motivate our interest in the objects of Theorem \ref{thm:MET}
through two applications. The aim here is to relate various objects
from these applications with the eigen-quantities, especially $P_{\star}$,
which will later show how to approximate using a particle algorithm.
Each subsection contains a different application and can be read separately.

\subsection{Importance sampling for tail probabilities\label{sub:Rare-events-estimation}}

For a measurable function $U:\mathsf{X}\rightarrow[-1,1]$ which is
not constant $\nu-a.e.$, some $\delta\in\left(0,1\right)$ and $m\geq1$,
our objective is to estimate the deviation probability 
\begin{equation}
\pi_{m}(\delta):=\mathbb{P}_{x}\left(\sum_{p=1}^{m}U(X_{p})>m\delta\right),\label{eq:dev_prob}
\end{equation}
where $\mathbb{P}_{x}$ denotes the law of $(X_{n};n\geq0)$ as a
Markov chain with $X_{0}=x$ and $X_{n}\sim M(X_{n-1},\cdot)$. There
is a quite extensive literature on methods for estimating probabilities
of the form (\ref{eq:dev_prob}) (see for example \citep{bucklew1990monte,dupuis2005dynamic}\textcolor{black}{,)
building upon large deviation theory for functionals of Markov chains,
with the results in \citep{iscoe1985large,ney1987markov} being particularly
relevant in the present context. We will explore an importance sampling
scenario in the setting of \citet{bucklew1990monte}. The choice of
this setup and specific form of }$\pi_{m}\left(\delta\right)$\textcolor{black}{{}
provides some insight into the applicability of the proposed algorithm,
but many of the details could be generalized. }

\textcolor{black}{For $\alpha\in\mathbb{R}$, introduce} 
\[
G_{\alpha}(x):=e^{\alpha U(x)},\quad\quad Q_{\alpha}(x,dx^{\prime}):=G_{\alpha}(x)M(x,dx^{\prime}).
\]
Note that $Q_{\alpha}^{(n)}(x,\mathsf{X})=\mathbb{E}_{x}\left[\exp\left(\sum_{p=0}^{n-1}\alpha U(X_{p})\right)\right]$. 

To simplify the discussion, assume that $Q_{\alpha}$ satisfies \textbf{(H)}
for each $\alpha\in\mathbb{R}$, which implies $M$ is uniformly recurrent;
see Appendix \ref{sub:Proofs-Numellin} for a definition of recurrence
and related details. We denote by $h_{\star}^{\alpha}$, $\Lambda_{\star}(\alpha),$
$\eta_{\star}^{\alpha},P_{\star}^{\alpha}$ the eigen-quantities and
twisted kernel corresponding to $Q_{\alpha}$. It is then a consequence
of Theorem \ref{thm:MET} that 
\[
\Lambda_{\star}(\alpha)=\lim_{n\rightarrow\infty}\frac{1}{n}\log\mathbb{E}_{x}\left[\exp\left(\alpha\sum_{p=0}^{n-1}U(X_{p})\right)\right].
\]
The convex dual of $\Lambda_{\star}(\alpha)$ is
\begin{equation}
I(t):=\sup_{\alpha\in\mathbb{R}}\left[t\alpha-\Lambda_{\star}(\alpha)\right],\quad t\in\mathbb{R}.\label{eq:lambda_conjugate}
\end{equation}

\textcolor{black}{\citet{bucklew1990monte}} proposed to estimate
$\pi_{m}(\delta)$ by importance sampling, using some Markov kernel
$\overline{M}$ such that $M(x,\cdot)\ll\overline{M}(x,\cdot)$. For
$L\geq1$, we consider the estimator of $\pi_{m}(\delta)$:
\begin{equation}
\widehat{\pi}_{m}\left(\delta,L\right):=\frac{1}{L}\sum_{i=1}^{L}\mathbb{I}\left[\sum_{p=1}^{m}U(X_{p}^{i})>m\delta\right]\frac{\mathrm{d}\mathbb{P}_{x}}{\mathrm{d}\overline{\mathbb{P}}_{x}}(X_{0}^{i},...,X_{m}^{i}),\label{eq:mu_n_hat}
\end{equation}
where $\left\{ \left(X_{0}^{i},X_{1}^{i},...,X_{m}^{i}\right);i=1,...,L\right\} $
is composed by $L$ independent Markov chains, each with transition
kernel $\overline{M}$ and law denoted by $\overline{\mathbb{P}}_{x}$.
The corresponding expectation will be denoted below by $\overline{\mathbb{E}}_{x}$.
Note that the dependence of $\widehat{\pi}_{m}\left(\delta,L\right)$
on $\overline{M}$ is suppressed from the notation. Also following
\citep[Definition 2.]{bucklew1990monte} we will consider a class
of candidates for $\overline{M}$. Let $\mathcal{C}$ be the collection
of Markov transitions $\overline{M}$ for each of which there exists
$0<\bar{\epsilon}^{-},\bar{\epsilon}^{+}<\infty$ and a probability
measure $\bar{\nu}$ such that 
\[
\left(\mathcal{C}\right)\quad\quad\quad\bar{\nu}\left(\cdot\right)\bar{\epsilon}^{-}\leq\overline{M}(x,\cdot)\leq\bar{\epsilon}^{+}\bar{\nu}\left(\cdot\right),\;\forall x,\quad\quad\nu\ll\bar{\nu},\quad\quad\int\left(\frac{\mathrm{d}\nu}{\mathrm{d}\bar{\nu}}\left(x\right)\right)^{2}\bar{\nu}\left(dx\right)<\infty,
\]
where $\nu$ is as in \textbf{(H)}. 

The following result describes the asymptotic $m\rightarrow\infty$
behavior of the probability of interest and the second moment of the
estimator when $L=1$. 
\begin{thm}
\label{thm:bucklew}\citep{bucklew1990monte} \end{thm}
\begin{enumerate}
\item \emph{$I(t)$ is a non-negative, strictly convex function with $I(t)=0$
if and only if $t=\Lambda_{\star}^{\prime}(0)$.}
\item \emph{For any $\delta\in\left(0,1\right)$, the following large deviation
principle holds
\[
\lim_{m\rightarrow\infty}\frac{1}{m}\log\pi_{m}\left(\delta\right)=-\inf_{t\in[\delta,\infty)}I(t).
\]
}
\item \emph{For any $\delta\in\left(0,1\right)$ and $\overline{M}$ in
the class $\mathcal{C}$, the importance sampling estimator satisfies
\begin{equation}
\lim_{m\rightarrow\infty}\frac{1}{m}\log\overline{\mathbb{E}}_{x}\left[\widehat{\pi}_{m}\left(\delta,1\right)^{2}\right]\geq-2\inf_{t\in[\delta,\infty)}I(t).\label{eq:asymp_efficient}
\end{equation}
}
\item \emph{For any $\delta\in\left(0,1\right)$ and $\alpha$ the unique
solution of $\Lambda_{\star}^{\prime}\left(\alpha\right)=\delta$,
the twisted kernel $P_{\star}^{\alpha}$ is the unique member of the
class $\mathcal{C}$ for which equality holds in (\ref{eq:asymp_efficient}),
and as such is called} asymptotically efficient\emph{. }\end{enumerate}
\begin{proof}
We just point to the appropriate references. Parts 1.-3. are due to
\citet[Theorem 1 and Corollary 1]{bucklew1990monte}, in turn derived
from various results of \citep{iscoe1985large}. Equation (9) in \citep{bucklew1990monte}
is satisfied trivially in the present scenario since $I(t)$ is continuous.
Part 4. is an application of \citet[Theorem 3]{bucklew1990monte}.
We note that the authors there consider the kernel $M\left(x,dy\right)G_{\alpha}\left(y\right)$,
as opposed to $G_{\alpha}\left(x\right)M\left(x,dy\right)$, this
difference is of no consequence due to the asymptotic ($m\rightarrow\infty$)
nature of the results and the fact that the two corresponding twisted
kernels are essentially identical.
\end{proof}
The following elementary corollary summarizes an important practical
implication of this theorem. 
\begin{cor}
Assume $\inf_{t\in\left[\delta,\infty\right)}I(t)\neq0$. Unless $\overline{M}$
is chosen to be $P_{\star}^{\alpha}$ with $\alpha$ the solution
to $\Lambda_{\star}^{\prime}\left(\alpha\right)=\delta$, the number
of samples $L$ must increase at a strictly positive exponential rate
in $m$ in order to prevent growth of the relative variance: 
\begin{equation}
\overline{\mathbb{E}}_{x}\left[\left(\frac{\widehat{\pi}_{m}\left(\delta,L\right)}{\pi_{m}(\delta)}-1\right)^{2}\right]=\frac{1}{L}\left(\frac{\overline{\mathbb{E}}_{x}\left[\widehat{\pi}_{m}\left(\delta,1\right)^{2}\right]}{\pi_{m}(\delta)^{2}}-1\right),\label{eq:rare_rel_var_1}
\end{equation}
as $m\to\infty$. Note that $\overline{\mathbb{E}}_{x}[\widehat{\pi}_{m}\left(\delta,L\right)]=\pi_{m}(\delta)$,
so (\ref{eq:rare_rel_var_1}) is indeed the relative variance. 
\end{cor}

\subsection{Optimal control with Kullback-Leibler divergence costs\label{sub:Optimal-control-with}}

We consider a particular class of fully observable stochastic control
problems in discrete time. Let $(X_{n};n\geq0)$ be a controlled Markov
chain initialized from $X_{0}=x$ and $X_{n}\sim M^{f_{n-1}}(X_{n-1},\cdot)$.
Here for each $n\geq0$ $f_{n}\in\mathcal{H}:=\left\{ h:\mathsf{X}\rightarrow\mathbb{R}_{+}^{*};\quad0<M(h)(x)<\infty;\;\forall x\right\} $,
where the set $\mathcal{H}$ is called the set of admissible control
functions. We refer to the sequence of control functions, $f=(f_{0},f_{1},\ldots)$,
as the policy. We will denote the Kullback-Leibler divergence between
the controlled and control-free Markov kernels as:
\[
\mathcal{KL}\left(\left.M^{f_{p}}\right\Vert M\right)(x):=\int M^{f_{p}}(x,dy)\log\frac{dM^{f_{p}}(x,\cdot)}{dM(x,\cdot)}(y).
\]
Let $U,\varOmega\in\mathcal{L}$. We are interested to compute the
optimal policies for the following control problems:

\begin{equation}
\mbox{Finite Horizon Cost}\qquad V_{0}(x)=\inf_{f\in\mathsf{\mathcal{H}}^{n}}\mathbb{E}_{x,0}^{f}\left[\sum_{p=0}^{n-1}\left(U(X_{p})+\mathcal{KL}\left(\left.\check{M}^{f_{p}}\right\Vert M\right)(X_{p})\right)+\varOmega(X_{n})\right],\label{eq:finite_hor_cost}
\end{equation}

\begin{equation}
\mbox{Infinite Horizon Average Cost}\qquad V_{\star}(x)=\inf_{f\in\mathcal{H}^{\mathbb{N}}}\limsup_{n\rightarrow\infty}\frac{1}{n}\mathbb{E}_{x,0}^{f}\left[\sum_{p=0}^{n}\left(U(X_{p})+\mathcal{KL}\left(\left.M^{f_{p}}\right\Vert M\right)(X_{p})\right)\right],\label{eq:inf_hor_cost}
\end{equation}
where $\mathbb{E}_{x,p}^{f}$ denotes the expectation over the path
of the controlled chain starting at $X_{p}=x$, where $p<n$ and $n$
is a deterministic finite horizon time. The interpretation of (\ref{eq:finite_hor_cost})-(\ref{eq:inf_hor_cost})
is that $M$ specifies the desired ``natural'' or control free dynamics
of the state of some stochastic system. The controlled state evolves
according to the dynamics specified by $M^{f_{p}}$ and $\mathcal{KL}\left(\left.M^{f_{p}}\right\Vert M\right)$
penalizes the discrepancy between $M^{f_{p}}(x,\cdot)$ and $M(x,\cdot)$.
The term $U(x)$ expresses an arbitrary state dependent stage cost
and $\varOmega$ is the terminal stage cost for time $n$. It is also
possible to write discounted cost versions of (\ref{eq:inf_hor_cost})
or non-stationary cost versions of (\ref{eq:finite_hor_cost}), but
these possible extensions are omitted. 

This problem was first posed for the finite horizon case in \citep{albertini1988logarithmic}.
The authors in \citep{albertini1988logarithmic} used unpublished
work of Sheu to formulate a duality between non-linear filtering and
optimal control similar to earlier work for continuous time models
found in \citep{fleming1982optimal,fleming1982logarithmic,sheu1984stochastic}.
As a result, one can perform computations for the dual filtering and
smoothing problem and then recover the optimal policy and value functions.
Although the stage costs in (\ref{eq:finite_hor_cost})-(\ref{eq:inf_hor_cost})
might not seem very intuitive they do include Gaussian problems with
quadratic costs (see Example \ref{ex:lqg}) or popular containment
problems (see Section \ref{sec:Examples}). More recently, there has
also been a renewed interest in this type of problems from the machine
learning community \citep{todorov2008general,theodorou2010generalized,kappen2005linear,dvijotham2011unified,bierkens2011online}.
However, outside of situations like Example \ref{ex:lqg}, analytical
solutions are rarely available and so numerical approximations are
required. 
\begin{example}
\label{ex:lqg}Consider the scalar controlled Markov model, $X_{p}=a(X_{p-1})+\mathrm{u}_{p-1}+W_{p},$
with $a(\cdot)$ is bounded continuous non-linear function, $W_{p}$
is an independent zero mean Gaussian random variable with variance
$\sigma^{2}$ and $\mathrm{u}_{p}$ is a standard control input. For
the controlled kernel we write 
\[
M^{f_{p-1}}(x_{p-1},dx_{p})=\frac{1}{\sqrt{2\pi\sigma^{2}}}\exp\left(-\frac{1}{2\sigma^{2}}\left(x_{p}-a(x_{p-1})-\mathrm{u}_{p-1}\right)^{2}\right)dx_{p}.
\]
 In what follows, it will be convenient to think of $f_{p}$ as coming
from $M^{f_{p}}(x,dy)=\frac{M(x,dy)f_{p}(y)}{M(f_{p})(x)}$, as it
will turn out that the dynamic programming solution for this problem
takes this form. So in this example we will set $f_{p}(y)=\exp\left(\frac{y\mathrm{u}_{p}}{\sigma^{2}}-\frac{\mathrm{u}_{p}^{2}}{2\sigma^{2}}\right)$.
The control-free model is $X_{p}=a(X_{p-1})+W_{p},$ so for the uncontrolled
kernel we have $M=M^{0}.$ For the stage cost, let $U(x)=\frac{1}{2\sigma^{2}}x^{2}$
and we have $\mathcal{KL}\left(\left.M^{f_{p}}\right\Vert M\right)=\frac{\mathrm{u}_{p}^{2}}{2\sigma^{2}},$
so we recover the usual quadratic cost control problem.
\end{example}
We now present a useful lemma that will be used when manipulating
the dynamic programming recursions.
\begin{lem}
\label{lem:duality-meas-inf} (Gibbs variational inequality) For every
$\nu\in\mathcal{P}$, $\psi>0$ such that $\nu\left(e^{-\psi}\right)<\infty$,
we have $\log\nu\left(e^{-\psi}\right)=-\inf_{\mu\in\mathcal{C}(\nu)}\left\{ \mu(\psi)+\mathcal{KL}\left(\left.\mu\right\Vert \nu\right)\right\} $,
where $\mathcal{C}(\nu)=\left\{ \mu\in\mathcal{P}:\:\mu\ll\nu\right\} $.
Moreover the infimum is attained for $\mu^{\ast}$ such that $\frac{d\mu^{*}}{dv}=\frac{e^{-\psi}}{\nu\left(e^{-\psi}\right)}.$
\end{lem}
The proof is standard and omitted; see for instance \citep[Proposition 1.4.2]{dupuis2011weak}
or \citep{dai1996connections}. We proceed by looking at the finite
and infinite horizon case separately.

\subsubsection*{The finite horizon case}

For the problem in (\ref{eq:finite_hor_cost}) define the value functions
or optimal cost to go at every time time $0\leq p<n$: 
\begin{equation}
V_{p}(x):=\inf_{\left(f_{l}\in\mathcal{H};\:p<l<n\right)}\left\{ U(x)+\mathcal{KL}\left(\left.M^{f_{p}}\right\Vert M\right)(x)+\mathbb{E}_{x,p}^{f}\left[\sum_{l=p+1}^{n-1}\left(U(X_{l})+\mathcal{KL}\left(\left.M^{f_{l}}\right\Vert M\right)(X_{l})\right)+\varOmega(X_{n})\right]-\sum_{l=p}^{n}\Lambda_{l}\right\} ,\label{eq:value_functions}
\end{equation}
with $V_{n}=\varOmega$. Let $\left(f_{p}^{*};\:0\leq p<n\right)$
denote the corresponding minimizing control functions in (\ref{eq:value_functions}).
Compared to (\ref{eq:finite_hor_cost}), $\sum_{l=p}^{n}\Lambda_{l}$
is a scaling constant that does not affect the solution. The significance
of this offset will become clear when we choose $\lambda_{p}=e^{\Lambda_{p}}$.
We proceed with a dynamic programming result:
\begin{lem}
\label{lem:value-function}The value function for problem (\ref{eq:value_functions})
at each time $p=0,\ldots,n-1$ is given by 
\begin{equation}
V_{p}(x)=U(x)-\Lambda_{p}+\inf_{f_{p}\in\mathcal{H}}\left\{ \mathcal{KL}\left(\left.M^{f_{p}}\right\Vert M\right)(x)+M^{f_{p}}\left(V_{p+1}\right)(x)\right\} \label{eq:value_f}
\end{equation}
with $V_{n}=\varOmega$. Let $Q=e^{-U}M$, $\lambda_{p}=e^{-\Lambda_{p}}$.
In addition, for each $p<n$ we have $V_{p+1}=-\log h_{p}$,  where
$h_{p}$ is given by the following backward recursion: 
\begin{equation}
Q(h_{p+1})=\lambda_{p}h_{p}.\label{eq:eig-funct}
\end{equation}
Furthermore, the optimal control is given by $f_{p}^{*}=h_{p}$ and
the optimally controlled Markov transition kernel by  
\[
M^{f_{p}^{\star}}(x,dy):=\frac{M(x,dy)h_{p}(y)}{M(h_{p})(x)}.
\]
 \end{lem}
\begin{proof}
Equation (\ref{eq:value_f}) states the standard dynamic programming
recursion for finite horizon problems, e.g. \citep[Theorem 3.2.1 ]{hernandez1996discrete}.
Using (\ref{eq:value_f}) and Lemma \ref{lem:duality-meas-inf} we
obtain $V_{p}=U-\Lambda_{p}-\log M\left(\exp\left(-V_{p+1}\right)\right)$
that can be rewritten as $e^{-V_{p}-\Lambda_{p}}=e^{-U}M(e^{-V_{p+1}})$.
By setting $\lambda_{p}=e^{-\Lambda_{p}}$, $h_{p}=e^{-V_{p+1}}$
we get (\ref{eq:eig-funct}) and the second part of Lemma \ref{lem:duality-meas-inf}
can be invoked to show that the expression for $M^{f_{p}^{\star}}$
follows by direct substitution with the optimal control being $f_{p}^{*}=\exp(-V_{p+1})=h_{p}$. 
\end{proof}
Note that the optimal controls appear as a multiplicative ``twisting''
function of the uncontrolled Markov transition kernel $M$. In addition,
it is clear from this result is that the non-negative operator $Q$
is equivalent to a multiplicative dynamic programming operator. Although
the scaling provided by $\Lambda_{p}$ can be arbitrary, the particular
choice is convenient for using simulated samples from $\eta_{p}$
to approximate $V_{p},h_{p}$; details will be presented in Section
\ref{sec:Overview}.
\begin{rem}
\label{rem:log-transform}Lemma \ref{lem:value-function} provides
an interpretation of $h_{p}$ as a log transform of a value function
similar to \citep{albertini1988logarithmic}. The similarity between
$h_{p}$ and $M^{f_{p}^{*}}$ with $h_{p,n}$ and $P_{(p,n)}$ is
clear. Despite this, we have purposely used a different notation for
$h_{p}$ and $h_{p,n}$, due to initializing with $h_{n}=\exp\left(-\Omega\right)$. 
\end{rem}

\subsubsection*{The infinite horizon case and interpretation of $h_{\star}$ and
$P_{\star}$}

We will look now at the infinite horizon average cost problem of (\ref{eq:inf_hor_cost}).
The objective is: (a) to compute a solution $\left(V_{\star},\varsigma_{\star}\right)$
of the Bellman average-cost optimality equation:
\begin{equation}
V_{\star}(x)+\varsigma_{\star}=\inf_{h\in\mathcal{H}}\left[U(x)+\mathcal{KL}\left(\left.M^{h}\right\Vert M\right)(x)+M^{h}\left(V_{\star}\right)(x)\right],\label{eq:bellman}
\end{equation}
\textcolor{black}{where $V_{\star}$ is the optimal value function
and $\varsigma_{\star}$ is the infinite horizon optimal average cost,
and (b) to compute $h_{\star}$, where $h_{\star}$ is the minimizer
for the infimum in (\ref{eq:bellman}). }Note that for this type of
problem the optimal policy can be shown to be stationary, i.e. the
optimal control functions is the same for every time $p$; \textcolor{black}{see
}\citep[Chapter 5]{hernandez1996discrete}\textcolor{black}{{} for background
and details. We relate now }(\ref{eq:bellman}) with the eigen-problem.
\begin{prop}
The average-cost Bellman equation (\ref{eq:bellman}) is satisfied
with $V_{\star}(x)=-\log h_{\star}(x),$ $\varsigma_{\star}=-\log\lambda_{\star}$,
where $\lambda_{\star},h_{\star}$ are the principal eigen-pair corresponding
to $Q:=e^{-U}M$. Furthermore the infimum in (\ref{eq:bellman}) is
achieved by taking $h=h_{\star}$ and the corresponding optimally
controlled dynamics evolve according to $P_{\star}.$\end{prop}
\begin{proof}
Applying Lemma \ref{lem:duality-meas-inf} and taking log's shows
that $\left(V_{\star},\varsigma_{\star}\right)$ is a solution of
the Bellman equation (\ref{eq:bellman}) if and only if 
\begin{equation}
V_{\star}(x)+\varsigma_{\star}=U(x)-\log M\left(e^{-V_{\star}}\right)(x),\label{eq:bellman_eigen}
\end{equation}
which is a re-writing of $Q(h_{\star})=\lambda_{\star}h_{\star},$
if $\varsigma_{\star}=-\log\lambda_{\star}$ and $V_{\star}=-\log h_{\star}$.
For establishing that $P_{\star}$ gives indeed the optimally controlled
dynamics we use again the second part of Lemma \ref{lem:duality-meas-inf}
and observe that the minimizer in (\ref{eq:bellman}) is attained
for $h=h_{\star}$. \end{proof}
\begin{rem}
In view of Proposition \ref{prop:_h_n_h_star_bound}, one may view
the backward recursion $h_{p,n}\left(x\right)=\frac{Q\left(h_{p+1,n}\right)}{\lambda_{p}}$
as a value iteration procedure, which aims to approximate $V_{\star}$
as $-\log h_{p,n}$ with $n$ being a finite horizon truncation used
for numerical purposes. 
\end{rem}

\section{Particle approximations for principal eigen-functions and related
quantities\label{sec:Overview}}

We propose a method to approximate the various eigen-quantities Algorithm
\ref{alg:pf}. The algorithm consists of a forward-backward recursion
approximating the deterministic quantities presented in Section \ref{sec:Deterministic_approximations_intro}.
A more precise probabilistic specification of the algorithm is given
in Section \ref{sec:Particle-approximations} and in Sections \ref{sec:Error-bounds},
\ref{sub:conditional_simulation_P^N} we present our convergence results.
The proofs not shown in Section \ref{sec:Overview} can be found in
the Appendix.

\subsection{The particle algorithm}

Algorithm \ref{alg:pf} has parameters: $N$, the particle population
size; $n$, the (half) time-horizon; and $\mu$, an initial probability
distribution. As we shall see, the values of $N$ and $n$ influence
the accuracy of the approximation and the choice of $\mu$ turns out
to be somewhat unimportant. 

\begin{algorithm}[H]
\uline{Forward recursion}

\qquad{}\quad{}Initialization: 

\qquad{}\qquad{}Sample $(\zeta_{0}^{i})_{i=1}^{N}\iid\mu$, 

\qquad{}\quad{}For $p=1,...,2n$, :

\qquad{}\qquad{}Sample $\left.(\zeta_{p}^{i})_{i=1}^{N}\right|(\zeta_{p-1}^{i})_{i=1}^{N}\quad\iid\quad\dfrac{\sum_{j=1}^{N}G(\zeta_{p-1}^{j})M(\zeta_{p-1}^{j},\cdot)}{\sum_{j=1}^{N}G(\zeta_{p-1}^{j})}.$ 

\uline{Backward recursion}

\qquad{}\quad{}Initialization: 

\qquad{}\qquad{}Set $h_{2n,2n}(x)=1,\quad x\in\mathsf{X}$

\qquad{}\quad{}For $p=2n-1,...,n$, :

\qquad{}\qquad{}\quad{}Set $h_{p,2n}^{N}(x)={\displaystyle \sum_{j=1}^{N}\frac{q(x,\zeta_{p+1}^{j})}{\sum_{i=1}^{N}q(\zeta_{p}^{i},\zeta_{p+1}^{j})}h_{p+1,2n}^{N}(\zeta_{p+1}^{j}).\quad x\in\mathsf{X}}$

\caption{Particle method for computing principal eigen-quantities\label{alg:pf}}
\end{algorithm}

We will take the random function $h_{n,2n}^{N}$ as an approximation
of $h_{\star}$ and the random kernel
\begin{equation}
P_{(n,2n)}^{N}(x,dx^{\prime}):=\frac{1}{h_{n-1,2n}^{N}(x)}\sum_{j=1}^{N}\frac{q(x,\zeta_{n}^{j})}{\sum_{i=1}^{N}q(\zeta_{n-1}^{i},\zeta_{n}^{j})}h_{n,2n}^{N}(\zeta_{n}^{j})\delta_{\zeta_{n}^{j}}\left(dx^{\prime}\right).\label{eq:intro_P_N_defn}
\end{equation}
as an approximation of $P_{\star}$. Note that, if so desired, each
$h_{p,2n}^{N}$ appearing in the algorithm can be evaluated at any
point $x\in\mathsf{X}$, but each step of the backward recursion actually
requires evaluation of $h_{p+1,2n}^{N}$ \emph{only} on the random
grid $\left\{ \zeta_{p+1}^{i};i=1,...,N\right\} $. Further note the
subscripting in $P_{(n,2n)}^{N}$ is not the semigroup index notation
of (\ref{eq:semi_g_notation}), and pertains only to the particular
kernel in (\ref{eq:intro_P_N_defn}). Occurrences will be kept to
an absolute minimum.

\subsection{Properties of the particle approximations\label{sec:Particle-approximations}}

We now provide a probabilistic specification of the quantities in
Algorithm \ref{alg:pf} and present some of their key properties,
which will be used to obtain $L_{r}$ bounds on the errors $h_{n,2n}^{N}(x)-h_{\star}(x)$
and $P_{(n,2n)}^{N}(x,A)-P_{\star}(x,A)$ (in terms of $N$ and $n$)
in Section \ref{sec:Error-bounds} and an unbiasedness result when
$\left(P_{(p,2n)}^{N};\:p>n\right)$ is used as an importance sampling
proposal in Section \ref{sub:conditional_simulation_P^N}.

\subsubsection*{Preliminaries}

For $N\geq1$, the particle system in the forward part of the algorithm
can be constructed as a canonical Markov chain with sample space $\Omega_{N}:=\left(\mathsf{X}^{N}\right)^{\mathbb{N}}$,
endowed with the corresponding product $\sigma$-algebra, derived
from the underlying $\sigma$-algebra $\mathcal{B}$. The state of
the chain at time $n\geq0$ is the $n$-th coordinate projection of
$\omega\in\Omega_{N}$ denoted by $\zeta_{n}(\omega)=\left(\zeta_{n}^{1}(\omega),\ldots,\zeta_{n}^{N}(\omega)\right)$,
taking values in $\mathsf{X}^{N}$. The natural filtration is denoted
by $\mathcal{F}_{n}=\sigma(\zeta_{0},\cdots,\zeta_{n})$, where the
dependence of each $\zeta_{n}$ and $\mathcal{F}_{n}$ on $N$ is
suppressed from the notation. 

We introduce collections of random probability measures $(\eta_{n}^{N})_{n\geq0}$:
\[
\eta_{n}^{N}:=\frac{1}{N}\sum_{i=1}^{N}\delta_{\zeta_{n}^{i}},\quad n\geq0.
\]
The law of the $N$-particle system is denoted by $\mathbb{P}_{N}$,
and in integral form, the initial distribution and transition probabilities
of the process $(\zeta_{n})_{n\geq0}$ are given by 
\begin{eqnarray}
\mathbb{P}_{N}(\mbox{\ensuremath{\zeta}}_{0}\in dx_{0}) & = & \prod_{i=1}^{N}\mu(dx_{0}^{i})\nonumber \\
\mathbb{P}_{N}(\left.\mbox{\ensuremath{\zeta}}_{n}\in dx_{n}\right|\mbox{\ensuremath{\zeta}}_{n-1}) & = & \prod_{i=1}^{N}\frac{\eta_{n-1}^{N}Q(dx_{n}^{i})}{\eta_{n-1}^{N}Q(1)}=\prod_{i=1}^{N}\Phi(\eta_{n-1}^{N})(dx_{n}^{i}),\quad n\geq1,\label{eq:particle transitions}
\end{eqnarray}
where $dx_{n}$ is an infinitesimal neighborhood of $x_{n}=\left(x_{n}^{1},\ldots x_{n}^{N}\right)\in\mathsf{X}^{N}.$
The expectation corresponding to $\mathbb{P}_{N}$ is denoted $\mathbb{E}_{N}$.

The idea for the eigen-function approximation in the algorithm is
to consider the identity
\begin{eqnarray}
h_{p-1,n}(x) & = & \frac{1}{\lambda_{p-1}}\int Q(x,dy)h_{p,n}(y)\nonumber \\
 & = & \frac{1}{\lambda_{p-1}}\int\frac{\mathrm{d}Q(x,\cdot)}{\mathrm{d}\eta_{p}}(y)h_{p,n}(y)\eta_{p}(dy)\nonumber \\
 & = & \frac{1}{\lambda_{p-1}}\int\frac{dQ(x,\cdot)}{d\Phi(\eta_{p-1})}(y)h_{p,n}(y)\eta_{p}(dy)\nonumber \\
 & = & \int\frac{\mathrm{d}Q(x,\cdot)}{\mathrm{d}(\eta_{p-1}Q)}(y)h_{p,n}(y)\eta_{p}(dy),\label{eq:h_recurse_explain}
\end{eqnarray}
where the first equality is due to the definition of the functions
$(h_{p,n})$, the second equality is just a change of measure in the
integral, and the third and fourth equalities are due to $\eta_{p}(\cdot)=\Phi(\eta_{p-1})(\cdot)=\dfrac{\eta_{p-1}Q(\cdot)}{\eta_{p-1}(G)}$
and the definition $\lambda_{p-1}=\eta_{p-1}(G)$. For any $x$ and
$p$, the derivative $\frac{\mathrm{d}Q(x,\cdot)}{\mathrm{d}\eta_{p}}$
is well defined under \textbf{(H)} because $Q(x,\cdot)$ is then equivalent
to $\nu$ for any $x$, and then also equivalent to $\eta_{p}$. 

Loosely speaking, the backward recursion of the algorithm arises from
taking the random measures $(\eta_{p}^{N})$ in place of $(\eta_{p})$
in (\ref{eq:h_recurse_explain}). To be more precise, let $\left(Q_{n}^{N}\right)$
be the collection of random integral kernels defined by 
\begin{equation}
Q_{n}^{N}(x,dx^{\prime}):=\frac{\mathrm{d}Q(x,\cdot)}{\mathrm{d}\Phi\left(\eta_{n-1}^{N}\right)}(x^{\prime})\eta_{n}^{N}(dx^{\prime}),\quad n\geq1.\label{eq:Q^N_defn}
\end{equation}
It is convenient to recall the semigroup notation in this context:
\[
Q_{n,n}^{N}:=Id,\quad\quad Q_{p,n}^{N}:=Q_{p+1}^{N}\cdots Q_{n}^{N},\quad p<n.
\]
Now define
\begin{equation}
\lambda_{n}^{N}:=\eta_{n}^{N}(G),\quad n\geq0,\label{eq:lambda^N_defn}
\end{equation}
and mimicking (\ref{eq:h_p_n_defn}) let $\left(h_{p,n}^{N}\right)$
be the collection of random functions defined by
\begin{equation}
h_{n,n}^{N}(x):=1,\quad\quad h_{p,n}^{N}(x):=\frac{Q_{p,n}^{N}(1)(x)}{\eta_{p}^{N}Q_{p,n}^{N}(1)},\quad0\leq p<n.\label{eq:h_pn_defn}
\end{equation}
Also, generalizing from the definition of $P_{(p,2n)}^{N}$ in (\ref{eq:intro_P_N_defn}),
define 
\[
P_{(p,n)}^{N}(x,dx^{\prime}):=\frac{Q_{p}^{N}(x,dx^{\prime})h_{p,n}^{N}(x^{\prime})}{\lambda_{p-1}^{N}h_{p-1,n}^{N}(x)}.
\]
The following lemma establishes relationships between these objects
which may be considered stochastic counterparts of the relations of
Lemma \ref{lem:approx_eigen}. 
\begin{lem}
\label{lem:_random_eigen_equations}The random measures $\left(\eta_{n}^{N}\right)$,
functions $\left(h_{p,n}^{N}\right),$ and kernels $\left(Q_{n}^{N}\right)$
satisfy 
\begin{equation}
\eta_{p}^{N}Q_{p+1}^{N}=\lambda_{p}^{N}\eta_{p+1}^{N},\quad\quad Q_{p+1}^{N}(h_{p+1,n}^{N})=\lambda_{p}^{N}h_{p,n}^{N},\quad\quad\eta_{p}^{N}(h_{p,n}^{N})=1,\quad0\leq p<n.\label{eq:particle_eig_eq}
\end{equation}
\begin{equation}
\eta_{p}^{N}Q_{p,n}^{N}(1)=\prod_{\ell=p}^{n-1}\lambda_{\ell}^{N},\quad\quad0\leq p<n.\label{eq:particle_op_prod_formula}
\end{equation}
\end{lem}
\begin{proof}
For the measure equation in (\ref{eq:particle_eig_eq}) and the definitions
(\ref{eq:Q^N_defn})-(\ref{eq:lambda^N_defn}),
\begin{eqnarray}
\eta_{p}^{N}Q_{p+1}^{N}(dx^{\prime}) & = & \eta_{p+1}^{N}(dx^{\prime})\int\eta_{p}^{N}(dx)\frac{dQ(x\cdot)}{d\Phi\left(\eta_{p}^{N}\right)}(x^{\prime})\nonumber \\
 & = & \lambda_{p}^{N}\eta_{p+1}^{N}(dx^{\prime})\int\eta_{p}^{N}(dx)\frac{q(x,x^{\prime})}{\int\eta_{p}^{N}(dy)q(y,x^{\prime})}\nonumber \\
 & = & \lambda_{p}^{N}\eta_{p+1}^{N}(dx^{\prime}).\label{eq:meas_eqn_N_proof}
\end{eqnarray}
By iterated application of (\ref{eq:meas_eqn_N_proof}) we have
\[
\eta_{p}^{N}Q_{p,n}^{N}(1)=\lambda_{p}^{N}\eta_{p+1}^{N}Q_{p+1,n}^{N}(1)=\left(\prod_{\ell=p}^{n-1}\lambda_{\ell}^{N}\right)\eta_{n}^{N}Q_{n,n}^{N}(1)=\prod_{\ell=p}^{n-1}\lambda_{\ell}^{N},
\]
where the final equality is due to the convention $Q_{n,n}^{N}:=Id$.
This establishes (\ref{eq:particle_op_prod_formula}). For the function
equation in (\ref{eq:particle_eig_eq}), we have 
\begin{eqnarray*}
Q_{p+1}^{N}\left(h_{p+1,n}^{N}\right) & = & \frac{Q_{p,n}^{N}(1)}{\eta_{p+1}^{N}Q_{p+1,n}^{N}(1)}\\
 & = & \lambda_{p}^{N}h_{p,n}^{N},
\end{eqnarray*}
where the final inequality holds due to (\ref{eq:particle_op_prod_formula}).
The right-most equality in (\ref{eq:particle_eig_eq}) holds directly
from the definition of $h_{p,n}^{N}$.\end{proof}
\begin{rem}
The recursion in the ``backward'' part of the algorithm is a re-arrangement
of the middle equation in (\ref{eq:particle_eig_eq}).
\end{rem}

\subsubsection*{Lack of bias}

Next we will see how iterates of the random operators $\left(Q_{p}^{N}\right)$
can be used to obtain unbiased estimates of iterates of the underlying
operator $Q$. 
\begin{prop}
\label{prop:unbias}Fix $N\geq1$ arbitrarily. Let $\mu^{\prime}\in\mathcal{P}$
and let $\mu^{N}$ be an $\mathcal{F}_{0}$-measurable random measure
satisfying $\mathbb{E}_{N}\left[\mu^{N}\left(A\right)\right]=\mu^{\prime}\left(A\right)$
for all $A\in\mathcal{B}$. Then for any $\varphi\in\mathcal{L}$
and $n\geq0$
\[
\mathbb{E}_{N}\left[\mu^{N}Q_{0,n}^{N}\left(\varphi\right)\right]=\mu^{\prime}Q^{\left(n\right)}\left(\varphi\right).
\]

\end{prop}

\begin{rem}
We highlight two interesting instances of initial measures in Proposition
\ref{prop:unbias}. The first is the degenerate case in which $\mu^{N}=\mu^{\prime}$,
for some $\mu^{\prime}\in\mathcal{P}$ other than $\mu$: in this
case we note that there is no bias (in the sense that the Proposition
\ref{prop:unbias} holds) when the functional $\mu^{N}Q_{0,n}^{N}\left(\varphi\right)$
involves a deterministic initial measure, \emph{other} than that used
to initialize the particle system. The second case is that in which
$\mu^{\prime}=\mu$ and $\mu^{N}=\eta_{0}^{N}$. In this case we have
\begin{eqnarray*}
\eta_{0}^{N}Q_{0,n}^{N}\left(\varphi\right) & = & \eta_{0}^{N}(G)\int\int\eta_{0}^{N}\left(dx_{0}\right)\frac{dQ\left(x_{0},\cdot\right)}{d\eta_{0}^{N}Q}\left(x_{1}\right)Q_{1,n}^{N}\left(\varphi\right)\left(x_{1}\right)\eta_{1}^{N}\left(dx_{1}\right)\\
 & = & \eta_{0}^{N}(G)\int\int\eta_{0}^{N}\left(dx_{0}\right)\frac{q\left(x_{0},x_{1}\right)}{{\displaystyle \frac{1}{N}\sum_{i=1}^{N}q\left(\zeta_{0}^{i},x_{1}\right)}}Q_{1,n}^{N}\left(\varphi\right)\left(x_{1}\right)\eta_{1}^{N}\left(dx_{1}\right)\\
 & = & \eta_{0}^{N}(G)\int Q_{1,n}^{N}\left(\varphi\right)\left(x_{1}\right)\eta_{1}^{N}\left(dx_{1}\right)\\
 & = & \prod_{p=0}^{n-1}\eta_{p}^{N}(G)\eta_{n}^{N}\left(\varphi\right),
\end{eqnarray*}
where the final equality can be verified by a simple induction. So
in this case, we recover from Proposition \ref{prop:unbias} the equality
$\mathbb{E}_{N}\left[\prod_{p=0}^{n-1}\eta_{p}^{N}(G)\eta_{n}^{N}\left(\varphi\right)\right]=\mu Q^{\left(n\right)}\left(\varphi\right)$,
which is well known for the ``forward'' part of the particle algorithm
\citep[Chapter 9]{smc:theory:Dm04}.
\end{rem}

\begin{rem}
A number of generalizations of Proposition \ref{prop:unbias} may
be obtained quite directly. Consider some integral kernel $\widetilde{Q}$
different from $Q$ and which, for simplicity, satisfies $\widetilde{Q}(x,\cdot)\ll Q(x,\cdot)$
for all $x$. Then defining 
\[
\widetilde{Q}_{n}^{N}(x,dx^{\prime}):=\frac{d\widetilde{Q}(x,\cdot)}{d\Phi\left(\eta_{n-1}^{N}\right)}(x^{\prime})\eta_{n}^{N}(dx^{\prime}),\quad n\geq1,
\]
one can establish by similar arguments to those in the proof of Proposition
\ref{prop:unbias} that 
\[
\mathbb{E}_{N}\left[\mu^{N}\widetilde{Q}_{0,n}^{N}\left(\varphi\right)\right]=\mu^{\prime}\widetilde{Q}^{\left(n\right)}\left(\varphi\right),\quad n\geq0,
\]
i.e. that the particle system defining $\left(\eta_{n}^{N}\right)$
and whose law involves $Q$ can be used to obtain unbiased estimates
of product formulae involving $\widetilde{Q}$. In turn, this might
be of interest both in the present context and in other applications
of particle systems, when the aim is to approximate ratios of the
form
\[
\frac{\mu^{\prime}\widetilde{Q}^{\left(n\right)}\left(1\right)}{\mu Q^{\left(n\right)}\left(1\right)},
\]
although further details are beyond the scope of the present work.
The time-homogeneity can also easily be relaxed, of course under appropriate
domination assumptions. 
\end{rem}

\subsubsection*{Path-wise stability of the random operators}

Next we establish a sample path result for the random (and generally
path-wise inhomogeneous) semigroups $Q_{0,n}^{N}$ and $\frac{\mu^{\prime}Q_{0,n}^{N}}{\mu^{\prime}Q_{0,n}^{N}\left(1\right)}$,
where we show exponential stability uniformly with respect to $N$.
\begin{thm}
\textup{\label{thm_pathwise}The following path-wise, uniform bounds
hold for the random operators $\left(Q_{n}^{N}\right)$ and the corresponding
non-linear semigroup. For any $n\geq1$ and $\varphi\in\mathcal{L}$,
\begin{equation}
\sup_{\mu^{\prime}\in\mathcal{P}}\sup_{N\geq1}\sup_{\omega\in\Omega_{N}}\left|\left(\prod_{p=0}^{n-1}\lambda_{p}^{N}\right)^{-1}\mu^{\prime}Q_{0,n}^{N}\left(\varphi\right)-\mu^{\prime}\left(h_{0.n}^{N}\right)\eta_{n}^{N}\left(\varphi\right)\right|\left(\omega\right)\leq2\left\Vert \varphi\right\Vert \tilde{\rho}^{n}\left(\frac{\epsilon^{+}}{\epsilon^{-}}\right),\label{eq:particle_path_met}
\end{equation}
}

\begin{equation}
\sup_{\mu^{\prime}\in\mathcal{P}}\sup_{N\geq1}\sup_{\omega\in\Omega_{N}}\left|\frac{\mu^{\prime}Q_{0,n}^{N}\left(\varphi\right)}{\mu^{\prime}Q_{0,n}^{N}\left(1\right)}-\eta_{n}^{N}\left(\varphi\right)\right|\left(\omega\right)\leq2\left\Vert \varphi\right\Vert \tilde{\rho}^{n}\left(\frac{\epsilon^{+}}{\epsilon^{-}}\right)^{2},\label{eq:particle_path_filter_stab}
\end{equation}
where $\tilde{\rho}=1-\left(\epsilon^{-}/\epsilon^{+}\right)^{2}$.
\end{thm}
This type of uniform path-wise convergence plays an important role
in proving $L_{r}$ bounds that follows below.

\subsection{$L_{r}$ error estimates\label{sec:Error-bounds}}

The forward part of the algorithm has been suggested by \citet{smc:the:dMM03,smc:the:DMD04}
in order to approximate $\eta_{\star}$ and $\lambda_{\star}$ using
the empirical probability measures $\left(\eta_{n}^{N}\right)$. Defining
\begin{equation}
\Lambda_{n}^{N}:=\frac{1}{n}\sum_{p=0}^{n-1}\log\lambda_{p}^{N},\label{eq:Lambda_N}
\end{equation}
they proved estimates of the form 
\begin{eqnarray*}
\mathbb{E}_{N}\left[\left|\eta_{n}^{N}\left(\varphi\right)-\eta_{\star}\left(\varphi\right)\right|^{r}\right]^{1/r} & \leq & \left\Vert \varphi\right\Vert C\left(\frac{B_{r}}{\sqrt{N}}+\tilde{\rho}^{n}\right)\\
\mathbb{E}_{N}\left[\left|\Lambda_{n}^{N}-\Lambda_{\star}\right|^{r}\right]^{1/r} & \leq & C\left(\frac{B_{r}}{\sqrt{N}}+\frac{1}{n}\right)
\end{eqnarray*}
for some constants $C<\infty$ and $\tilde{\rho}<1$; see the final
expressions in the proofs of Theorem 2 and Corollary 2 of \citep{smc:the:DMD04}
for precise details. 
\begin{rem}
\citet{smc:the:DMD04} addressed the case that the function $G$ may
vanish, and a weaker ``multi-step'' version of \textbf{(H)}. Similar
techniques as used therein can be applied in the present context,
but involve notational complications.
\end{rem}
The backward recursion of Algorithm \ref{alg:pf} is relevant to the
main aim of this paper, i.e. to quantify the error in approximations
of $h_{\star,}$ and $P_{\star}$. This is presented in the following
result.
\begin{thm}
\label{thm:L_r_bounds}For any $r\geq1$ there is a universal constant
$B_{r}$ such that for any $n\geq1$, $0\leq p<n$ and $N\geq1$,
\begin{equation}
\sup_{x\in\mathsf{X}}\mathbb{E}_{N}\left[\left|h_{p,n}^{N}(x)-h_{\star}(x)\right|^{r}\right]^{1/r}\leq2\frac{B_{r}}{\sqrt{N}}\tilde{C}+C_{h}\rho^{p\wedge(n-p)},\label{eq:h_n_Lp}
\end{equation}
\begin{equation}
\sup_{x\in\mathsf{X}}\sup_{A\in\mathcal{B}}\mathbb{E}_{N}\left[\left|P_{(p,n)}^{N}\left(x,A\right)-P_{\star}\left(x,A\right)\right|^{r}\right]^{1/r}\leq4\frac{B_{r}}{\sqrt{N}}\tilde{C}\frac{\epsilon^{+}}{\epsilon^{-}}+C_{P}\rho^{p\wedge(n-p)}.\label{eq:P_n_Lp}
\end{equation}
where $\tilde{C}=\left[3\left(\frac{\epsilon^{+}}{\epsilon^{-}}\right)^{7}+\left(\frac{\epsilon^{+}}{\epsilon^{-}}\right)^{5}\frac{1}{1-\tilde{\rho}}\right]$
and $\rho,C_{h},C_{P}$ are as in Proposition \ref{prop:_h_n_h_star_bound}.
\end{thm}
The errors are thus controlled in $N$ , $p$ and $n$, and in these
bounds there is no dependence on the measure $\mu$ used in the initialization
of the algorithm. The proof uses the following decompositions
\[
h_{p,n}^{N}(x)-h_{\star}(x)=\frac{Q_{p+1}^{N}(h_{p+1,n}^{N})(x)}{\lambda_{p}^{N}}-\frac{Q(h_{p+1,n})(x)}{\lambda_{p}}+h_{p,n}(x)-h_{\star}(x),
\]
 and 
\[
P_{(p,n)}^{N}\left(x,A\right)-P_{\star}\left(x,A\right)=\Xi_{1}(x,A)+\Xi_{2}(x,A)+\Xi_{3}(x,A),
\]
where

\begin{eqnarray*}
\Xi_{1}(x,A) & := & \frac{1}{h_{p-1,n}^{N}(x)}\left[\frac{Q_{p}^{N}(h_{p,n}^{N}\mathbb{I}_{A})(x)}{\lambda_{p-1}^{N}}-\frac{Q(h_{p,n}\mathbb{I}_{A})(x)}{\lambda_{p-1}}\right]\\
\Xi_{2}(x,A) & := & \frac{Q(h_{p,n}\mathbb{I}_{A})(x)}{\lambda_{p-1}}\left[\frac{1}{h_{p-1,n}^{N}(x)}-\frac{1}{h_{p-1,n}(x)}\right]\\
\Xi_{3}(x,A) & := & P_{(p,n)}(x,A)-P_{\star}(x,A).
\end{eqnarray*}
Hence, it is crucial to provide additional $L_{r}$ bounds for $\frac{Q_{p}^{N}\left(\varphi h_{p,n}^{N}\right)(x)}{\lambda_{p-1}^{N}}-\frac{Q\left(\varphi h_{p,n}\right)(x)}{\lambda_{p-1}}$
for any $\varphi\in\mathcal{L}$. This is achieved in Proposition
\ref{prop:intermediate_L_r} (in the Appendix), but is based on cumbersome
expressions so more details are not presented here. 
\begin{rem}
The type of recursion in the backward part of the algorithm is implicitly
present (albeit expressed somewhat differently) in other interacting
particle algorithms, see for example \citep{del2010backward} and
\citep{douc2010sequential} in the context of non-linear filtering/smoothing
or \citet{del2011robustness,del2012snell} in the context of optimal
stopping problems. The main novelty of the present work stems from
finding the connection between the backward recursion and $h_{\star}$,
$P_{\star}$ and incorporating it in the analysis. Note also that
the forward part of the algorithm runs from $0$ up to $2n$, but
the backward part runs from $2n$ to $n$. 
\end{rem}

\subsection{Lack of bias and a $\chi^{2}$-distance bound for importance sampling
using\textmd{\normalsize{} $P_{(p,n)}^{N}\left(x,A\right)$\label{sub:conditional_simulation_P^N}}}

Section \ref{sub:Rare-events-estimation} showed an application where
one is interested to sample from $P_{\star}$ in the context of importance
sampling. Similarly, the twisted kernel approximations $(P_{p,n}^{N})_{p\leq n}$
can be used to achieve unbiased estimates of expectations on the path
space of the Markov process evolving with kernel $M$. One may use
the twisted kernel approximations after the forward-backward pass
of Algorithm \ref{alg:pf} and define an additional \emph{conditional}
simulation forward pass by sampling $X_{p}\sim P_{(n+p,2n)}^{N}(X_{p-1},\cdot),$
$p=1,\ldots,m$. When this simulation is used in the context of importance
sampling, a lack of bias result similar to Proposition \ref{prop:unbias}
follows. 
\begin{prop}
\label{prop:twsited_sampling}Fix $N\geq1$, $n\geq1$, $m\leq n$
and $x\in\mathsf{X}$ arbitrarily. Conditional on $\mathcal{F}_{2n}$,
let $(X_{p};p=0,...,m)$ be a non-homogeneous Markov chain with transitions
\begin{equation}
X_{0}=x,\quad\quad X_{p}\sim P_{(n+p,2n)}^{N}(X_{p-1},\cdot),\quad p=1,\ldots,m,\label{eq:twisted_sampling}
\end{equation}
where $\left(P_{\left(n+p,2n\right)}^{N}\right)$ are obtained from
Algorithm \ref{alg:pf}. Let $\mathbb{E}_{N}$ denote the expectation
w.r.t. the joint law of the particle system and $(X_{p})$ sampled
according to (\ref{eq:twisted_sampling}). Then, for any integrable
function $F:\mathsf{X}^{m+1}\rightarrow\mathbb{R}$, 
\begin{equation}
\mathbb{E}_{N}\left[F(X_{0:m})\frac{h_{n,2n}^{N}(X_{0})}{h_{n+m,2n}^{N}(X_{m})}\prod_{p=0}^{m-1}\frac{\lambda_{n+p}^{N}}{G(X_{p})}\right]=\mathbb{E}_{x}\left[F(X_{0:m})\right],\label{eq:IS_estimator_particle}
\end{equation}
where on the r.h.s. $\mathbb{E}_{x}$ denotes expectation w.r.t. the
law of a Markov chain $(X_{p};p=0,...,m)$ with $X_{0}=x$ and $X_{p}\sim M(X_{p-1},\cdot)$. 
\end{prop}
We can also quantify the discrepancy between the law of $(X_{p};p=0,...,m)$
when obtained from (\ref{eq:twisted_sampling}), i.e. 
\[
\mbox{\ensuremath{\overline{\mathbb{P}}}}_{x}^{N,n}(X_{0}\in A_{0},\ldots,X_{m}\in A_{m}):=\mathbb{E}_{N}\left[\mathbb{I}[X_{0}\in A_{0},\ldots,X_{m}\in A_{m}]\right]
\]
and the ``ideal'' law:
\[
\mbox{\ensuremath{\overline{\mathbb{P}}}}_{x}(X_{0}\in A_{0},\ldots,X_{m}\in A_{m}):=\int_{A_{0}\times\cdots\times A_{m}}\delta_{x}(dx_{0})\prod_{p=1}^{m}P_{\star}(x_{p-1},dx_{p}).
\]
Indeed, since 
\[
\mathbb{P}_{x}(X_{0}\in A_{0},\ldots,X_{m}\in A_{m})=\int_{A_{0}\times\cdots\times A_{m}}\delta_{x}(dx_{0})\prod_{p=1}^{m}M(x_{p-1},dx_{p})=\mathbb{E}_{x}\left[\mathbb{I}[X_{0}\in A_{0},\ldots,X_{m}\in A_{m}]\right],
\]
it follows from (\ref{eq:IS_estimator_particle}) that up to null
sets, 
\[
\frac{\mathrm{d}\mathbb{P}_{x}}{\mathrm{d}\overline{\mathbb{P}}_{x}^{N,n}}(X_{0},\ldots,X_{m})=\mathbb{E}_{N}\left[\left.\frac{h_{n,2n}^{N}(X_{0})}{h_{n+m,2n}^{N}(X_{m})}\prod_{p=0}^{m-1}\frac{\lambda_{n+p}^{N}}{G(X_{p})}\right|X_{0},\ldots,X_{m}\right],
\]
and from the definition of $P_{\star}$ in (\ref{eq:twisted_intro}),
\[
\frac{\mathrm{d}\mathbb{P}_{x}}{\mathrm{d}\overline{\mathbb{P}}_{x}}(X_{0},\ldots,X_{m})=\frac{h_{\star}(X_{0})}{h_{\star}(X_{m})}\prod_{p=0}^{m-1}\frac{\lambda_{\star}}{G(X_{p})}.
\]
Therefore
\[
\frac{\mathrm{d}\overline{\mathbb{P}}_{x}}{\mathrm{d}\overline{\mathbb{P}}_{x}^{N,n}}(X_{0},\ldots,X_{m})=\mathbb{E}_{N}\left[\left.\frac{h_{n,2n}^{N}(X_{0})}{h_{\star}(X_{0})}\frac{h_{\star}(X_{m})}{h_{n+m,2n}^{N}(X_{m})}\prod_{p=0}^{m-1}\frac{\lambda_{n+p}^{N}}{\lambda_{\star}}\right|X_{0},\ldots,X_{m}\right].
\]
The following proposition estimates the $\chi^{2}$-distance (variance
of Radon-Nikodym derivative) between the two measures in question.
Restricting our attention to the case where the state space $\mathsf{X}$
is a finite set allows for a fairly straightforward proof, given in
the Appendix.
\begin{prop}
\label{prop:chi_square_bound} Assume that $\mathsf{X}$ is a finite
set and that the assumptions of Proposition \ref{prop:twsited_sampling}
hold. Then, there exists a finite constant $C$ depending on $\epsilon^{+},\epsilon^{-}$
such that the following bound holds for any $x\in\mathsf{X}$, $1\leq m\leq n$
and $N\geq1$, 
\begin{eqnarray}
 &  & \mathbb{E}_{N}\left[\left(\frac{\mathrm{d}\overline{\mathbb{P}}_{x}}{\mathrm{d}\overline{\mathbb{P}}_{x}^{N,n}}(X_{0},\ldots,X_{m})-1\right)^{2}\right]^{1/2}\nonumber \\
 &  & \leq C\left(1+\frac{C}{\sqrt{N}}\right)^{1/2}\left[\left(1+\frac{C}{N}\right)^{m}-1\right]^{1/2}+C\left[\frac{1}{\sqrt{N}}+\left(1-\frac{\epsilon^{-}}{\epsilon^{+}}\right)^{n-m}\right]\mathrm{card}(\mathsf{X}).\label{eq:chi_square_bound}
\end{eqnarray}

\end{prop}

\section{Numerical Examples\label{sec:Examples} }

We will present numerical examples for each application of Section
\ref{sec:Applications}.

\subsection{Importance Sampling for tail probabilities\label{sub:Importance-Sampling-for}}

We commence by this revisiting the problem in Section \ref{sub:Rare-events-estimation}
where the eigen-quantities arise from a rare-event estimation problem.
Recall we consider a Markov process starting from $x\in\mathsf{X}$
with transition kernel $M$ and are interested to estimate the tail
probability $\pi_{m}(\delta):=\mathbb{P}_{x}\left(\sum_{p=1}^{m}U(X_{p})>m\delta\right)$.
Following the results in Section \ref{sub:Rare-events-estimation}
we will choose $\overline{M}=P_{\star}^{\alpha}$ as the importance
kernel, where $\alpha$ is the unique solution of \emph{of $\Lambda_{\star}^{\prime}\left(\alpha\right)=\delta$}.
Then, the importance sampling estimate of $\pi_{m}(\delta)$ written
earlier in (\ref{eq:mu_n_hat}) becomes 
\begin{equation}
\widehat{\pi}_{m}\left(\delta,L\right)=\frac{1}{L}\sum_{i=1}^{L}\left(\mathbb{I}\left[\sum_{p=1}^{m}U(X_{p}^{i})>m\delta\right]\frac{\exp\left[m\Lambda_{\star}\left(\alpha\right)\right]}{\prod_{p=0}^{m-1}G_{\alpha}(X_{p}^{i})}\frac{h_{\star}^{\alpha}(X_{0}^{i})}{h_{\star}^{\alpha}(X_{m}^{i})}\right).\label{eq:pi_twisted}
\end{equation}
As per Proposition \ref{prop:twsited_sampling}, it is in fact possible
to achieve unbiased estimates using the twisted kernel approximations
to define a \emph{conditional} simulation distribution, and using
an estimator which mimics the form of (\ref{eq:pi_twisted}). 

It is an immediate corollary of Proposition \ref{prop:twsited_sampling}
that $\mathbb{E}_{N}[\widehat{\pi}_{m}\left(\delta,L\right)]=\pi_{m}(\delta)$,
and Proposition \ref{prop:chi_square_bound} indicates that r.h.s.
of (\ref{eq:chi_square_bound}) goes to zero as $m\to\infty$ if$N,n$
grow such that $m=o(n)$ and $m=o(N)$.

\subsubsection*{Numerics}

For some $c>0$ we take $\mathsf{X}=[-c,c]$ and consider an ergodic
Gaussian transition kernel with support restricted to $[-c,c]$, 
\[
M(x,dy)=\frac{\exp\left(-\frac{1}{2}\left(y-\frac{x}{2}\right)^{2}\right)}{\left(\mbox{erf}\left(\frac{c-x/2}{\sqrt{2}}\right)-\mbox{erf}\left(\frac{-c-x/2}{\sqrt{2}}\right)\right)\sqrt{2\pi}}\mathbb{I}_{[-c,c]}(y)dy,
\]
and consider $U$ defined by
\[
U\left(x\right)=\begin{cases}
-1 & \quad x\leq-1\\
x & \quad x\in(-1,1)\\
1 & \quad x\geq1.
\end{cases}
\]
For any $\alpha\in\mathbb{R}$, assumption \textbf{(H)} holds. The
left plot in Figure \ref{fig:rare_events1-1} shows estimated values
of $\pi_{m}(\delta)$, obtained from the algorithm with $N=250$,
$n=500$, $\alpha=6$ and using the estimator which appears inside
the expectation in \ref{eq:pi_twisted}, i.e. a single sample of the
conditional Markov chain. The displayed results are the averages over
$2000$ realizations of this entire procedure. The exponential decay
rate predicted by the large deviation principle (Theorem \ref{thm:bucklew},
part 2.) is apparent. The sample relative variances in the case of
$\delta=0.9$ are shown on the right of \ref{fig:rare_events1-1},
for different values of $\alpha$. The sample relative variance of
$\widehat{\pi}_{m}\left(0.9,1\right)$ for the trivial case $\overline{M}=M$
is also included for reference, and explodes rapidly with $m$. 

On a very fine grid of $\alpha$-values, approximations of $\Lambda_{\star}\left(\alpha\right)$
as per (\ref{eq:Lambda_N}) were obtained with the same settings of
$N$ and $n$. These were used to obtain the approximations of $\left[\alpha t-\Lambda_{\star}(\alpha)\right]$
against $\alpha$ plotted on the left of Figure \ref{fig:rare_events1}
and an approximation of $\Lambda_{\star}^{\prime}(\alpha)$ was obtained
by finite differences, the result is shown on the right of Figure
\ref{fig:rare_events1}. The latter plot suggests $\Lambda_{\star}^{\prime}\left(10\right)\approx0.9$,
and bearing in mind the optimality result of Theorem \ref{thm:bucklew},
part 4., we then notice in the relative variance plots of Figure \ref{fig:rare_events1-1}
that the slowest growth (amongst the $\alpha$ values considered)
occurs with $\alpha=8$. 

\begin{figure}
\centering\includegraphics[width=0.49\textwidth]{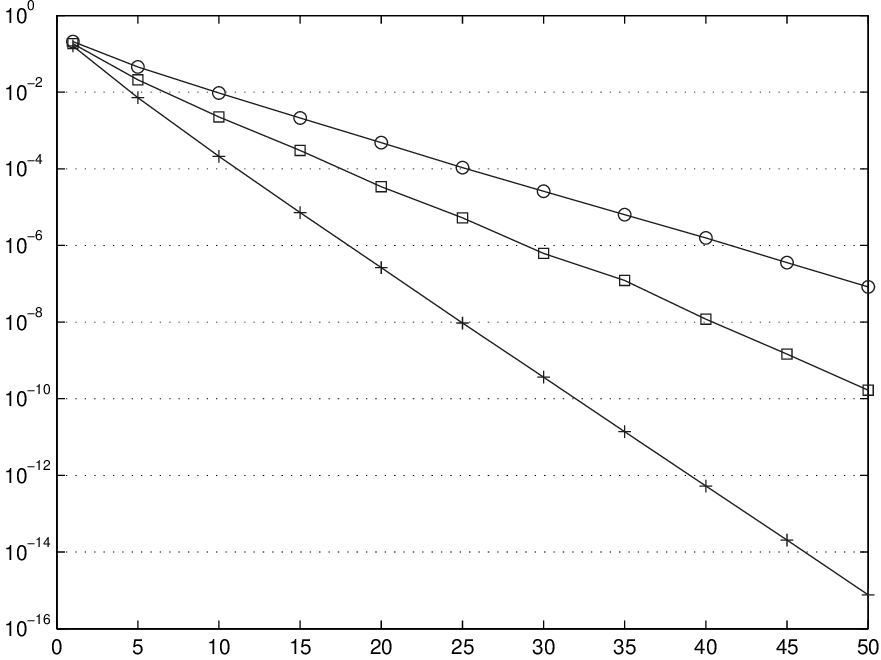}\hfill{}\includegraphics[width=0.49\textwidth]{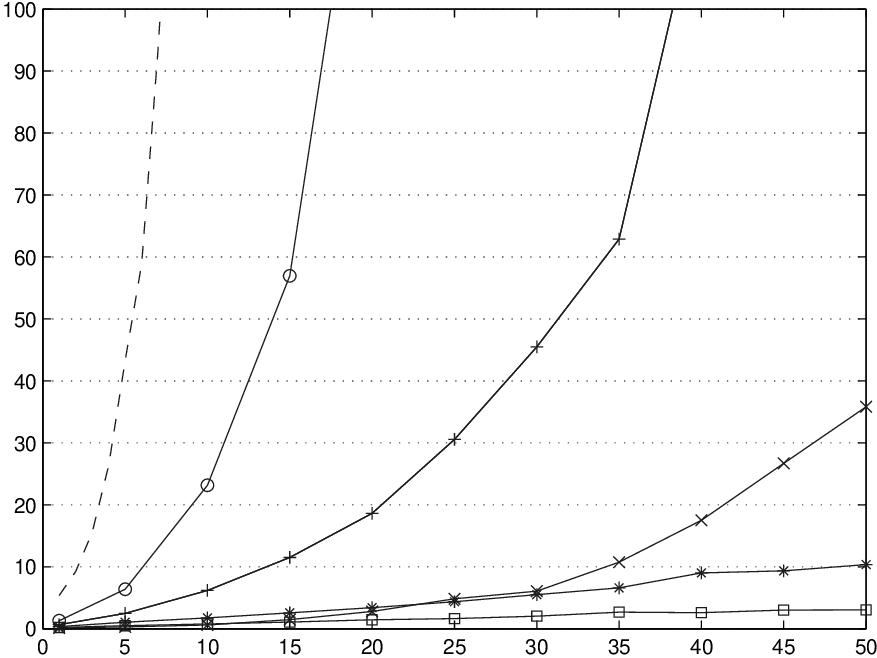}

\caption{Left: estimated value of $\pi_{m}(\delta)$ against $m$, for: $\circ$,$\delta=0.8$;
$\square$,$\delta=0.9$, and $+,$ $\delta=0.99$. Right: solid lines
show sample relative variance of the estimated value of $\pi_{m}(0.9)$
against $m$ using the conditional simulation method with: $\circ,$
$\alpha=1$; $+,\alpha=2$; $*$, $\alpha=4$; $\square$, $\alpha=8$;
and $\times,$ $\alpha=16$. Dashed line shows sample relative variance
of $\widehat{\pi}_{m}\left(0.9,1\right)$ in the case $\overline{M}=M$.}
\label{fig:rare_events1-1}
\end{figure}

\begin{figure}
\centering\includegraphics[width=0.49\textwidth]{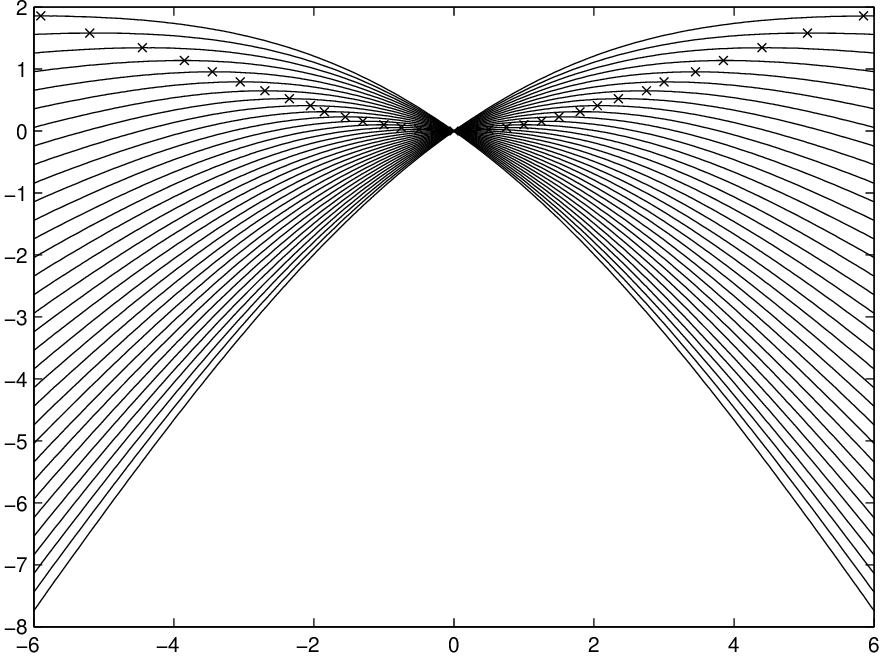}\hfill{}\includegraphics[width=0.49\textwidth]{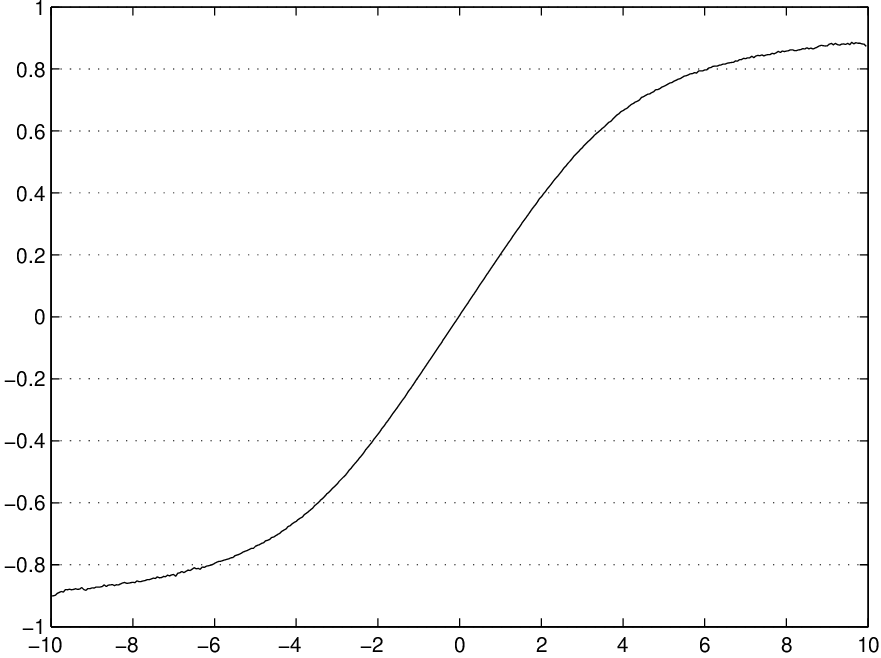}

\caption{Left: each of the solid curves shows an approximation of $\left[\alpha t-\Lambda_{\star}(\alpha)\right]$
against $\alpha$, with each curve corresponding to a different value
of $t$ in the range $\left[-0.8,0.8\right]$. The cross on each curve
indicates its maximum and thus approximates the value of $\sup_{\alpha}\left[\alpha t-\Lambda_{\star}(\alpha)\right]=I(t)$.
Right: $\Lambda_{\star}^{\prime}(\alpha)$ against $\alpha$ approximated
using finite differences. }
\label{fig:rare_events1}
\end{figure}

\subsection{Optimal control with $\mathcal{KL}$ stage costs\label{sub:Solving-a-Bellman}}

We will show some numerical results related to the control problem
of Section \ref{sub:Optimal-control-with}. We will look at the finite
and infinite horizon case separately.

\subsubsection*{Finite Horizon}

We begin by looking at a particular case of Example \ref{ex:lqg}.
Let $\mathsf{X}=\mathbb{R}^{2}$ and consider the controlled dynamics
being 
\[
X_{p}=\left[\begin{array}{cc}
1 & \tau\\
0 & 1
\end{array}\right]X_{p-1}+\left[\begin{array}{cc}
\tau & \tau^{2}/2\\
0 & \tau
\end{array}\right]\left(W_{p}+F_{p}\right),
\]
where $p=1,\ldots,n$ and $W_{n}$ are independent zero mean Gaussian
random variables with covariance matrix $\sigma^{2}I$ and $F_{n}\in\mathbb{R}^{2}$
are the standard control inputs. Note in general $M$ cannot satisfy
\textbf{(H)}, but truncation (and suitable re-normalization) of $M$
to any bounded interval of $\mathsf{X}$ does allow \textbf{(H)} to
be satisfied. Let also the state-dependent part of the stage cost
be $U(x)=(1-\boldsymbol{\mathbb{I}}_{\left(-\delta,\delta\right)}(x\left(1\right)))$\textcolor{black}{{}
for some $\delta>0$.} \textcolor{black}{This type of cost penalizes
states outside $(-\delta,\delta)$ and can be a convenient choice
for various containment problems. For this example we will set }$X_{0}$
to be zero mean Gaussian random variables with covariance matrix $\left[\begin{array}{cc}
3 & 0\\
0 & 1
\end{array}\right]$\textcolor{black}{. In Figure \ref{fig:Estimated-value-functions}
we present estimated some value functions for $T=2n=20$, }$\tau=0.1$,
$\delta=0.5$ and $N=500$. Note that the displayed value function
estimates are obtained by averaging over $50$ independent multiple
runs as due to the high variance of the initial condition the estimates
$h_{p,2n}^{N}$ exhibit a significant amount of variance. Still some
errors are visible in the form or ripples due to using a small $N$.

\begin{figure}
\begin{minipage}[t]{1\columnwidth}%
\centering\includegraphics[width=0.8\textwidth]{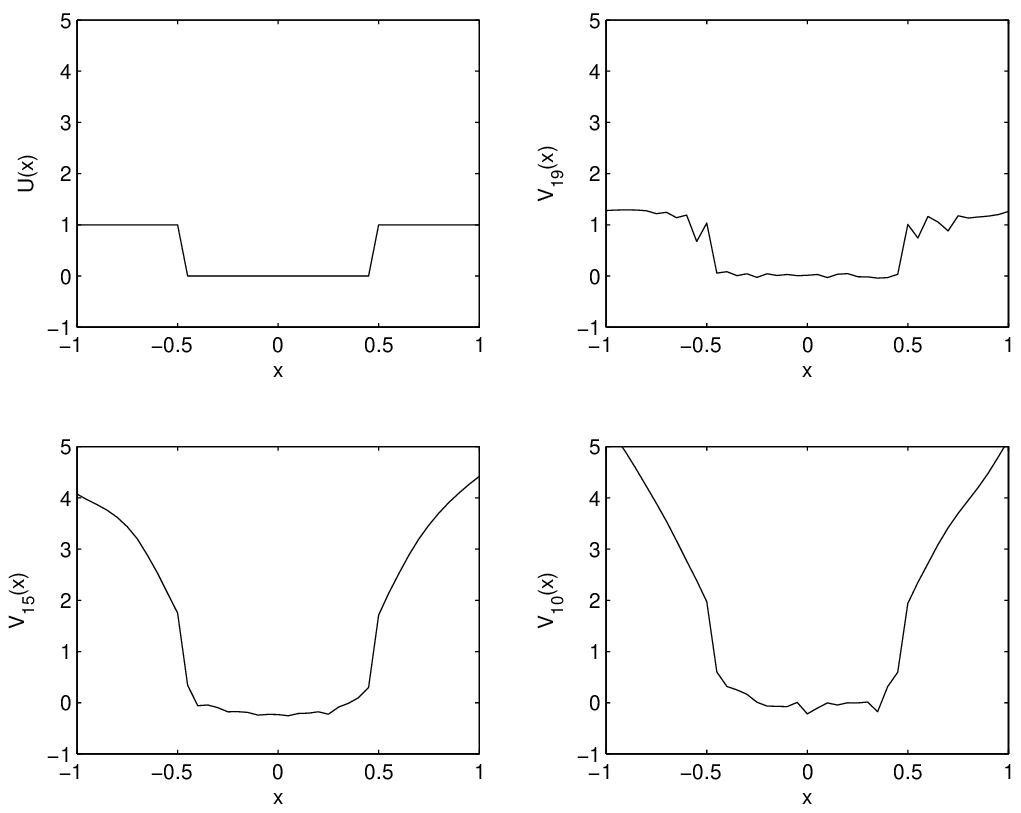}%
\end{minipage}\caption{\label{fig:Estimated-value-functions}Estimated value functions $V_{p}^{N}(x)=-\log h_{p,n}^{N}$
against $x$ for $p=10,15,19$ and $n=20$. Top left panel is $U(x)$
against $x$. }
\end{figure}

\subsubsection*{Infinite Horizon}

We will now look at a different infinite horizon scalar example. The
Cox-Ingersoll-Ross process satisfies 
\[
dX_{t}=\theta\left(\mu-X_{t}\right)dt+\sigma\sqrt{X_{t}}dW_{t},
\]
where $\left\{ W_{t}\right\} $ is standard one-dimensional Brownian
motion, $\theta>0$ is the reversion rate, $\mu>0$ is the level of
mean reversion and $\sigma>0$ specifies the volatility. In financial
applications this process is widely used to model interest rates.
When $2\theta\mu>\sigma^{2}$ it is stationary. Here $\mathsf{X}=\mathbb{R}^{+}$
and for purposes of illustration we consider the case that $M$ is
the transition probability from time $t=0$ to $t=0.01$ of the CIR
process, which is available in closed form \citep{cox1985theory}.
Although known to satisfy a type of multiplicative Lyapunov drift
condition which allows an MET to be established in a weighted $\infty$-norm
setting \citep{whiteley2012}, $M$ cannot satisfy \textbf{(H)}. Truncation
(and suitable re-normalization) of $M$ to any bounded interval of
$\mathsf{X}$ does allow \textbf{(H)} to be satisfied. In our numerical
experiments this truncation was made to $[0,500]$. We took the parameter
settings $\theta=2$, $\sigma=20,$ $\mu=10$ and considered, for
a range of $\delta$, the \textcolor{black}{following ``well-shaped}\textcolor{blue}{''}
cost function: 
\begin{equation}
U(x)=2\mathbb{I}_{[0,10-\delta]}(x)+\mathbb{I}_{\left[10+\delta,\infty\right)}(x),\label{eq:cost_function}
\end{equation}
\textcolor{black}{which penalizes states outside $(10-\delta,10+\delta)$.}

Figure \ref{fig:value_function} shows estimates of the value function,
which were obtained via averaging by evaluating the window-averaged
quantities $\frac{1}{m}\sum_{p=0}^{m-1}h_{n+p,2n}^{N}(x)$ with $N=500$,
$n=2000$ and $m=100$ and evaluations on a fine grid from $x=4$
to $x=20$. Note the coincidence of the discontinuities in (\ref{eq:cost_function})
with those in the estimated function. The influence of the parameter
$\delta$ is apparent. Table \ref{tab:variance} shows the empirical
relative varia\textcolor{black}{nce (variance over the square of the
mean) of the} estimated value function evaluations at different points
$x$ and for different numbers of particles $N$. The variance evidently
decreases with $N$, with large values associated with more extreme
values of $x$. 

\begin{figure}
\begin{minipage}[t]{1\columnwidth}%
\centering\includegraphics[width=0.6\textwidth]{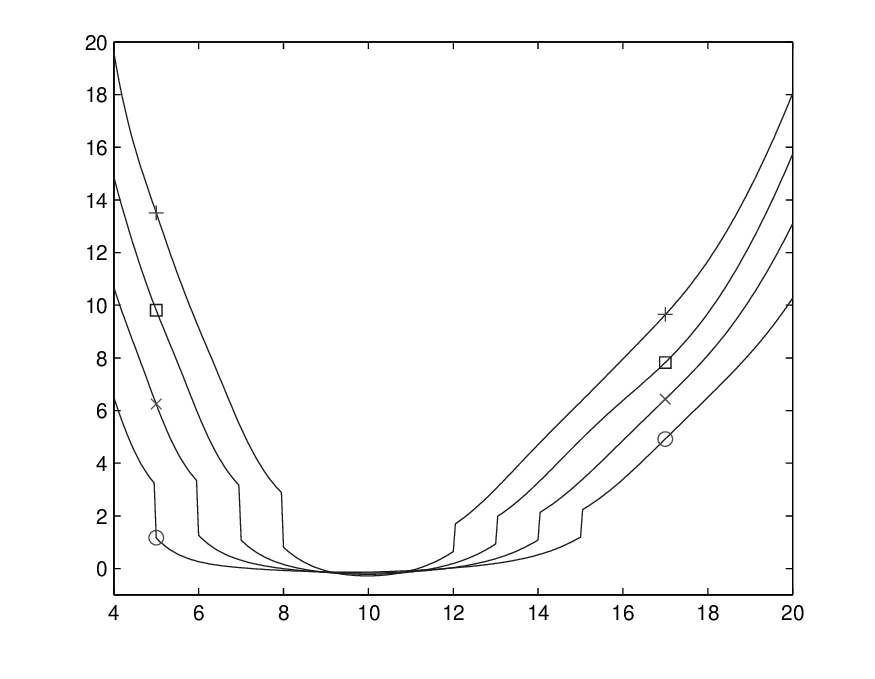}%
\end{minipage}\caption{\label{fig:value_function}Estimated optimal value function $V_{\star}(x)$
against $x$ for various parameter values: $\circ,\delta=5;\times,\delta=4;\square,\delta=3;+,\delta=2$.}
\end{figure}

\begin{table}
\begin{minipage}[t]{1\columnwidth}%
\centering%
\begin{tabular}{|c|c|c|c|c|c|c|}
\hline 
\multirow{2}{*}{N} & \multicolumn{6}{c|}{$x$}\tabularnewline
\cline{2-7} 
 & 6 & 8 & 10 & 12 & 14 & 16\tabularnewline
\hline 
50 & $1.81\times10^{-3}$ & $1.94\times10^{-5}$ & $5.62\times10^{-5}$ & $7.27\times10^{-5}$ & $1.07\times10^{-3}$ & $7.2\times10^{-3}$\tabularnewline
\hline 
100 & $1.02\times10^{-3}$ & $9.13\times10^{-6}$ & $2.78\times10^{-5}$ & $3.26\times10^{-5}$ & $5.41\times10^{-4}$ & $6.15\times10^{-3}$\tabularnewline
\hline 
500 & $1.15\times10^{-4}$ & $4.95\times10^{-6}$ & $1.46\times10^{-6}$ & $5.75\times10^{-6}$ & $3.08\times10^{-5}$ & $2.28\times10^{-3}$\tabularnewline
\hline 
\end{tabular}%
\end{minipage}\caption{\label{tab:variance}Empirical relative variance of value function
evaluations (at different $x$), with $n=2000$ from $500$ independent
realizations of the algorithm. }
\end{table}

\section{Discussion\label{sec:Extensions}}

We presented a generic particle algorithm to approximate the principle
eigen-function of an un-normalized positive Markov integral kernel
together with the associated twisted probability kernel. As per standard
Perron-Frobenius theory, we have not made any reversibility assumptions,
and this is reflected to some extent in the ``forward-backward''
structure of the algorithm. We also presented some theoretical results
demonstrating the validity of using such a numerical scheme and saw
how it can be applied to a variety of practical problems. 

There are a number of possible avenues for further investigation.
Regarding the theory, Assumption \textbf{(H) }is very restrictive
when $\mathsf{X}$ is non-compact. Starting points for the analysis
of the method under weaker assumptions are \citep{whiteley2011,whiteley2012},
where the stability of Feynman-Kac semigroups and particle approximations
have been studied under a relaxation of the uniform majorization/minorization
structure of \textbf{(H)}, using a Lyapunov drift condition. 

There also many aspects of the applications considered here that could
benefit from further study. The connection to optimal importance sampling
schemes for rare event simulation and estimation could be extended
by studying in detail the variance of the estimator appearing in Proposition
\ref{prop:twsited_sampling} as well as the propagation of chaos properties
associated with blocks of samples drawn from $\left(P_{(p,n)}^{N}\right)$.
Furthermore, it is of some interest to investigate how optimization
schemes such as those in \citep[Chapter 5]{kantas2009sequential}
could be combined with the algorithm in order to estimate the solution
of $\Lambda_{\star}^{\prime}(\alpha)=\delta$. Regarding this last
point, when the solution of $\alpha$ is not unique \citep{chan:lai}
by-pass the computation of the eigen-function using saddle-point approximations,
so it would be interesting to investigate how the two approaches could
be combined. Furthermore, the optimal control problem underlying the
Bellman equation in Section \ref{sub:Optimal-control-with} has only
recently received some mathematical attention \citep{theodorou2010generalized,dvijotham2011unified}
for the finite horizon case and could be investigated further. Especially
for the infinite horizon case, there are many connections with continuous
time control problems \citep{dai1996connections,sheu1984stochastic}
and further insight could extend the applicability of the numerical
tools in this paper.

\appendix

\section{Appendix}

\subsection{Proofs and auxiliary results for Section \ref{sec:Multiplicative-Ergodicity}
\label{sub:Proofs-Numellin}}

We now present dome definitions and preliminary results which preface
the proof of Theorem \ref{thm:MET} . The first is a lemma that establishes
uniform bounds on ratio functionals involving iterates of $Q$. Set
$\mathcal{L}^{+}:=\left\{ \varphi\in\mathcal{L}:\nu(\varphi)>0\right\} $. 
\begin{lem}
\label{lem:Q_n_bounds}For any $\mu^{\prime}\in\mathcal{P}$ and $\varphi\in\mathcal{L}^{+}$,
\begin{equation}
\inf_{n\geq1}\inf_{x\in\mathsf{X}}\frac{Q^{\left(n\right)}(\varphi)(x)}{\mu^{\prime}Q^{\left(n\right)}(\varphi)}\geq\frac{\epsilon^{-}}{\epsilon^{+}}>0,\quad\sup_{n\geq1}\sup_{x\in\mathsf{X}}\frac{Q^{\left(n\right)}(\varphi)(x)}{\mu^{\prime}Q^{\left(n\right)}(\varphi)}\leq\frac{\epsilon^{+}}{\epsilon^{-}}<\infty.\label{eq:in_sup_bound}
\end{equation}
\end{lem}
\begin{proof}
Under \textbf{(H)}, 
\[
\frac{Q^{\left(n\right)}\left(\varphi\right)(y)}{Q^{\left(n\right)}\left(\varphi\right)(x)}\leq\frac{\epsilon^{+}}{\epsilon^{-}}\quad\forall x,y\in\mathsf{X},n\geq1,
\]
then integrating in the numerator with respect to $\mu^{\prime}$
and re-arranging gives the infimum bound in (\ref{eq:in_sup_bound}).
The proof of the supremum bound is similar. 
\end{proof}
Following \citet{mc:the:N04}, the notions of irreducibility and aperiodicity
of a non-negative kernel generalize naturally from the probabilistic
case, and are expressed in terms of a $\sigma$-finite irreducibility
measure. For simplicity of presentation we shall take as this measure
the $\nu$ appearing in \textbf{(H)}. It follows immediately from
the definitions of \citep{mc:the:N04} that when \textbf{(H)} holds,
$Q$ is $\nu$-irreducible and aperiodic. The number $\lambda_{\star}$
as defined in (\ref{eq:Lambda})-(\ref{eq:lambda}) is called the
\emph{generalized principal eigen-value} (g.p.e.) of $Q$ by \citet[Theorem 3.1]{mc:theory:KM03}
and in our setting coincides with the reciprocal of the convergence
parameter of \citet[Section 3.2]{mc:the:N04}. 

Recall, the spectral radius of $Q$ as a bounded linear operator on
$\mathcal{L}$ is defined as $\xi:=\lim_{n\rightarrow\infty}\interleave Q^{\left(n\right)}\interleave^{1/n}$
(existence follows by sub-multiplicativity of operator norm). For
notational convenience define $s^{-}:\mathsf{X}\rightarrow\mathbb{R}_{+}$,
$s^{+}:\mathsf{X}\rightarrow\mathbb{R}_{+}$ by $s^{-}(x)=\epsilon^{-},s^{+}(x)=\epsilon^{+},\forall x$,
respectively. In the terminology \citet[Proposition 3.4]{mc:the:N04},
$Q$ is called $\lambda_{\star}$-\emph{recurrent} if and only if
$\sum_{n=0}^{\infty}\lambda_{\star}^{-n}\nu Q^{\left(n\right)}\left(s^{-}\right)=\infty$.
The following lemma prepares for Theorem \ref{thm:MET}.
\begin{lem}
\label{lem:recurrence}We have 
\begin{equation}
\epsilon^{-}\leq\xi=\lambda_{\star}\leq\epsilon^{+},\quad\quad\qquad\inf_{\mu^{\prime}\in\mathcal{P}}\inf_{n\geq0}\frac{\mu^{\prime}Q^{\left(n\right)}(1)}{\lambda_{\star}^{n}}>0,\label{eq:uniform_below}
\end{equation}
and therefore $Q$ is $\lambda_{\star}$ -recurrent.\end{lem}
\begin{rem}
Following the terminology and arguments of \citep[p.96]{mc:the:N04},
under \textbf{(H)} the kernel $Q$ is then additionally \emph{uniformly}
$\lambda_{\star}$-recurrent.\end{rem}
\begin{proof}
The upper and lower bounds on the spectral radius $\xi$ follow from
\textbf{(H)}, because for any $n\geq1$ and $x\in\mathsf{X}$ we have
$\epsilon^{-}\leq\left[Q^{(n)}(1)(x)\right]^{1/n}\leq\epsilon^{+}$.
To verify that $\lambda_{\star}$ coincides with $\xi$, write
\begin{eqnarray*}
\left|\frac{1}{n}\log\frac{\sup_{x}Q^{(n)}(1)(x)}{\nu Q^{(n)}(s^{-})}\right| & = & \left|\frac{1}{n}\log\frac{\sup_{x}Q^{(n)}(1)(x)}{\nu Q^{(n)}(1)}-\frac{1}{n}\log\epsilon^{-}\right|\\
 & \leq & \frac{1}{n}\log\frac{\epsilon^{+}}{\epsilon^{-}}+\frac{1}{n}\log\frac{\nu Q^{(n-1)}(1)}{\nu Q^{(n-1)}(1)}+\frac{1}{n}\left|\log\epsilon^{-}\right|\;\rightarrow\;0\text{\quad\ as\quad\ \ensuremath{n\rightarrow\infty}. }
\end{eqnarray*}
It remains to verify the uniform lower bound in (\ref{eq:uniform_below})
and thus the $\lambda_{\star}$-recurrence. A key feature of the majorization
part of assumption \textbf{(H)} is that it implies $\nu Q^{\left(n+m-1\right)}\left(s^{+}\right)\leq\nu Q^{\left(n-1\right)}\left(s^{+}\right)\nu Q^{\left(m-1\right)}\left(s^{+}\right)$
and then by \emph{sub}-additivity we are assured of the existence
of: 
\begin{equation}
\Lambda_{\star}^{+}:=\lim_{n\rightarrow\infty}\frac{1}{n}\log\nu Q^{\left(n-1\right)}\left(s^{+}\right)=\inf_{n\geq1}\frac{1}{n}\log\nu Q^{\left(n-1\right)}\left(s^{+}\right).\label{eq:Lambda_dash}
\end{equation}
But from the definitions of $s^{+}$ and $s^{-},$ 
\begin{eqnarray}
\frac{1}{n}\log\nu Q^{\left(n-1\right)}\left(s^{+}\right)-\frac{1}{n}\log\nu Q^{\left(n-1\right)}\left(s^{-}\right) & = & \frac{1}{n}\log\left[\frac{\nu Q^{\left(n-1\right)}(1)}{\nu Q^{\left(n-1\right)}(1)}\frac{\epsilon^{+}}{\epsilon^{-}}\right]\nonumber \\
 & = & \frac{1}{n}\log\left(\frac{\epsilon^{+}}{\epsilon^{-}}\right),\label{eq:squeeze}
\end{eqnarray}
so taking $n\rightarrow\infty$ we find that $\Lambda_{\star}^{+}=\Lambda_{\star}$,
and then (\ref{eq:squeeze}) together with the right-most equality
in (\ref{eq:Lambda_dash}) imply 
\[
\frac{1}{n}\log\nu Q^{\left(n-1\right)}\left(s^{-}\right)-\Lambda_{\star}\geq-\frac{1}{n}\log\left(\frac{\epsilon^{+}}{\epsilon^{-}}\right),
\]
so
\[
\frac{\nu Q^{\left(n-1\right)}\left(s^{-}\right)}{\lambda_{\star}^{n}}\geq\frac{\epsilon^{-}}{\epsilon^{+}}>0.
\]
Equation (\ref{eq:uniform_below}) then holds as $\frac{\mu'Q^{\left(n\right)}(1)}{\nu Q^{\left(n\right)}(1)\epsilon^{+}}\geq\frac{\epsilon^{-}}{\left(\epsilon^{+}\right)^{2}}$
for all $\mu'\in\mathcal{P}$ , and this implies $\lambda_{\star}$-recurrence.
\end{proof}
Now consider the family of potential kernels, $\left\{ U_{\theta};\theta\in[\lambda_{\star},\infty)\right\} $,

\[
U_{\theta}:=\sum_{n=0}^{\infty}\theta^{-n-1}\left(Q-s^{-}\otimes\nu\right)^{\left(n\right)}.
\]
where the convergence of the sum, in the operator norm, is ensured
by the $\lambda_{\star}$-recurrence of $Q$ (shown in Lemma \ref{lem:recurrence}
in Appendix) and is straightforward to verify using the inversion
argument of \citet[Proof of Lemma 3.2]{mc:theory:KM03}, noting that
as per Lemma \ref{lem:recurrence}, the spectral radius of $Q$ coincides
with the g.p.e., $\xi=\lambda_{\star}$.
\begin{proof}
(of Theorem \ref{thm:MET}) As per Lemma \ref{lem:recurrence}, the
spectral radius of $Q$ coincides with $\lambda_{\star}$. By the
same Lemma, $Q$ is $\lambda_{\star}-$recurrent. By \citep[Theorems 5.1 and 5.2]{mc:the:N04},
$\nu U_{\lambda_{\star}}$ and $U_{\lambda_{\star}}(s^{-})$ are then
respectively the unique measure and $\nu$-essentially unique non-zero
function satisfying 
\begin{equation}
\nu U_{\lambda_{\star}}Q=\lambda_{\star}\nu U_{\lambda_{\star}},\quad QU_{\lambda_{\star}}(s^{-})=\lambda_{\star}U_{\lambda_{\star}}(s^{-}),\quad\nu U_{\lambda_{\star}}\left(s^{-}\right)=1.\label{eq:unnorm_eigen}
\end{equation}
Under \textbf{(H)} we then have from (\ref{eq:unnorm_eigen}) that
\begin{equation}
0<\frac{\epsilon^{-}}{\lambda_{\star}}=\frac{\epsilon^{-}}{\lambda_{\star}}\nu U_{\lambda_{\star}}(s^{-})\leq U_{\lambda_{\star}}(s^{-})(x)\leq\frac{\epsilon^{+}}{\lambda_{\star}}\nu U_{\lambda_{\star}}(s^{-})=\frac{\epsilon^{+}}{\lambda_{\star}}<\infty,\quad\forall x,\label{eq:U_lam_bounds}
\end{equation}
thus we take 
\begin{equation}
\eta_{\star}:=\frac{\nu U_{\lambda_{\star}}}{\nu U_{\lambda_{\star}}\left(1\right)},\quad\quad h_{\star}:=\frac{U_{\lambda_{\star}}\left(s^{-}\right)}{\eta_{\star}U_{\lambda_{\star}}\left(s^{-}\right)}\label{eq:eigen_star_defns}
\end{equation}
 establishing (\ref{eq:eigen_star}). The uniqueness properties transfer
directly to $\eta_{\star}$ and $h_{\star}$. 

We obtain from (\ref{eq:unnorm_eigen}) and (\ref{eq:U_lam_bounds})
the following uniform lower and upper bounds on $h_{\star}$: 
\begin{equation}
h_{\star}(x)=\frac{Q\left(h_{\star}\right)(x)}{\lambda_{\star}}\geq\frac{\epsilon^{-}}{\lambda_{\star}}\nu\left(h_{\star}\right)=\frac{\epsilon^{-}}{\lambda_{\star}}\frac{\nu U_{\lambda_{\star}}\left(s^{-}\right)}{\eta_{\star}U_{\lambda_{\star}}\left(s^{-}\right)}=\frac{\epsilon^{-}}{\lambda_{\star}}\frac{1}{\eta_{\star}U_{\lambda_{\star}}\left(s^{-}\right)}\geq\frac{\epsilon^{-}}{\epsilon^{+}}>0,\quad\forall x,\label{eq:h_star_bound_below}
\end{equation}
\begin{equation}
h_{\star}(x)=\frac{Q\left(h_{\star}\right)(x)}{\lambda_{\star}}\leq\frac{\epsilon^{+}}{\lambda_{\star}}\nu\left(h_{\star}\right)=\frac{\epsilon^{+}}{\lambda_{\star}}\frac{1}{\eta_{\star}U_{\lambda_{\star}}\left(s^{-}\right)}\leq\frac{\epsilon^{+}}{\epsilon^{-}}<\infty,\quad\forall x\label{eq:h_star_bounded_above}
\end{equation}
so that (\ref{eq:eta_in_P_h_bound}) is established. Furthermore $P_{\star}$
is then well-defined as a Markov kernel and we readily verify that
it satisfies a uniform minorization condition:
\begin{eqnarray*}
P_{\star}(x,dx^{\prime}) & = & \frac{Q(x,dx^{\prime})h_{\star}(x^{\prime})}{h_{\star}(x)\lambda_{\star}}\\
 & \geq & \frac{\nu(h_{\star})}{h_{\star}(x)\lambda_{\star}}\frac{\epsilon^{-}\nu(dx^{\prime})h_{\star}(x^{\prime})}{\nu(h_{\star})}\\
 & = & \frac{1}{U_{\lambda_{\star}}\left(s^{-}\right)(x)\lambda_{\star}}\epsilon^{-}\nu(dx^{\prime})U_{\lambda_{\star}}\left(s^{-}\right)(x')\\
 & \geq & \frac{\epsilon^{-}}{\epsilon^{+}}\nu(dx^{\prime})U_{\lambda_{\star}}\left(s^{-}\right)(x'),\quad\forall x,
\end{eqnarray*}
where $\nu U_{\lambda_{\star}}\left(s^{-}\right)=1$ and (\ref{eq:U_lam_bounds})
have been used. Thus $P_{\star}$ is uniformly geometrically ergodic
and by inspection of the eigen-measure equation its unique invariant
probability distribution, denoted by $\pi_{\star}$, is given by $\pi_{\star}\left(\varphi\right)=\eta_{\star}\left(h_{\star}\varphi\right)/\eta_{\star}\left(h_{\star}\right)=\eta_{\star}\left(h_{\star}\varphi\right)$.
Then, again noting that $\nu U_{\lambda_{\star}}\left(s^{-}\right)=1$,
by \citep[Theorem 16.2.4]{mc:theory:MT09} we have:
\begin{equation}
\interleave P_{\star}^{\left(n\right)}-1\otimes\pi_{\star}\interleave\leq2\rho^{n},\label{eq:tiwsited_uniform_ergo}
\end{equation}
where $\rho:=1-\left(\epsilon^{-}/\epsilon^{+}\right)$, which establishes
(\ref{eq:P_star_ergo}). Multiplying by $h_{\star}>0$ in (\ref{eq:tiwsited_uniform_ergo})
yields for any $\phi\in\mathcal{L},x\in\mathsf{X}$, 
\begin{equation}
\left|\lambda_{\star}^{-n}Q^{\left(n\right)}\left(h_{\star}\phi\right)(x)-h_{\star}(x)\eta_{\star}\left(h_{\star}\phi\right)\right|\leq2\rho^{n}h_{\star}(x)\left\Vert \phi\right\Vert \leq2\rho^{n}\left(\frac{\epsilon^{+}}{\epsilon^{-}}\right)\left\Vert \phi\right\Vert ,\label{eq:pre_MEt}
\end{equation}
where (\ref{eq:h_star_bounded_above}) has been used. By equation
(\ref{eq:h_star_bound_below}), $h_{\star}$ is bounded below away
from zero and therefore for any $\varphi\in\mathcal{L}$, we may have
taken $\phi:=\varphi/h_{\star}\in\mathcal{L}$ in (\ref{eq:pre_MEt}).
Finally noting from (\ref{eq:h_star_bound_below}) that $\left\Vert \varphi/h_{\star}\right\Vert \leq\left(\epsilon^{+}/\epsilon^{-}\right)\left\Vert \varphi\right\Vert $,
the bound of (\ref{eq:MET_bound}) is established. 
\end{proof}

\subsection{Proofs and auxiliary results for Section \ref{sec:Deterministic_approximations_intro}}

Under assumption \textbf{(H)} we obtain uniform bounds on these quantities,
as per the following Lemma.
\begin{lem}
\label{lem:h_p_n_bounds}$\;$
\begin{equation}
\inf_{n\geq0}\eta_{n}(G)>0\label{eq:lambda_bound below}
\end{equation}
\begin{equation}
\inf_{n\geq1}\inf_{0\leq p\leq n}\inf_{x\in\mathsf{X}}h_{p,n}(x)\geq\frac{\epsilon^{-}}{\epsilon^{+}}>0,\quad\sup_{n\geq1}\sup_{0\leq p\leq n}\sup_{x\in\mathsf{X}}h_{p,n}(x)\leq\frac{\epsilon^{+}}{\epsilon^{-}}<\infty.\label{eq:h_bounded}
\end{equation}
\end{lem}
\begin{proof}
Assumption \textbf{(H)} implies that $G$ is bounded below away from
zero and therefore we have (\ref{eq:lambda_bound below}). Lemma \ref{lem:Q_n_bounds}
in the Appendix implies (\ref{eq:h_bounded}).
\end{proof}
We proceed with the proof of Proposition \ref{prop:_h_n_h_star_bound}:
\begin{proof}
(of Proposition \ref{prop:_h_n_h_star_bound}) We first treat (\ref{eq:eta_converge}),
\begin{eqnarray*}
\left\Vert \eta_{n}-\eta_{\star}\right\Vert  & = & \sup_{\varphi:\left|\varphi\right|\leq1}\left|\mu Q^{\left(n\right)}\left(\varphi\right)\left[\frac{1}{\mu Q^{\left(n\right)}\left(1\right)}-\frac{1}{\lambda_{\star}^{n}\mu(h_{\star})}\right]+\frac{\mu Q^{\left(n\right)}\left(\varphi\right)}{\lambda_{\star}^{n}\mu(h_{\star})}-\eta_{\star}\left(\varphi\right)\right|\\
 & \leq & \sup_{\varphi:\left|\varphi\right|\leq1}\left|\frac{\mu Q^{\left(n\right)}\left(\varphi\right)}{\mu Q^{\left(n\right)}\left(1\right)}\right|\left|\frac{\mu Q^{\left(n\right)}\left(1\right)}{\lambda_{\star}^{n}\mu(h_{\star})}-1\right|\\
 &  & +\sup_{\varphi:\left|\varphi\right|\leq1}\left|\frac{\mu Q^{\left(n\right)}\left(\varphi\right)}{\lambda_{\star}^{n}\mu(h_{\star})}-\eta_{\star}\left(\varphi\right)\right|\\
 & \leq & \frac{2}{\mu(h_{\star})}\rho^{n}\left(\frac{\epsilon^{+}}{\epsilon^{-}}\right)^{2}\\
 &  & +\frac{2}{\mu(h_{\star})}\rho^{n}\left(\frac{\epsilon^{+}}{\epsilon^{-}}\right)^{2}\\
 & \leq & 4\rho^{n}\left(\frac{\epsilon^{+}}{\epsilon^{-}}\right)^{3},
\end{eqnarray*}
where the penultimate inequality follows from two applications of
the bound of Theorem \ref{thm:MET}, Equation (\ref{eq:MET_bound}),
and the final inequality is due to (\ref{eq:eta_in_P_h_bound}). This
establishes (\ref{eq:eta_converge}). 

In order to prove (\ref{eq:h_converge}), we first consider products
of the values $\left(\lambda_{n}\right)$. We have

\begin{eqnarray}
\left|\frac{\prod_{\ell=p}^{n-1}\lambda_{\ell}}{\lambda_{\star}^{n-p}}-1\right| & = & \left|\frac{\eta_{p}Q^{\left(n-p\right)}(1)}{\lambda_{\star}^{n-p}}-\eta_{p}(h_{\star})+\eta_{p}(h_{\star})-\eta_{\star}(h_{\star})\right|\nonumber \\
 & \leq & \left|\frac{\eta_{p}Q^{\left(n-p\right)}(1)}{\lambda_{\star}^{n-p}}-\eta_{p}(h_{\star})\right|+\left|\eta_{p}(h_{\star})-\eta_{\star}(h_{\star})\right|\nonumber \\
 & \leq & 2\rho^{n-p}\left(\frac{\epsilon^{+}}{\epsilon^{-}}\right)^{2}+4\rho^{p}\left(\frac{\epsilon^{+}}{\epsilon^{-}}\right)^{3}\left\Vert h_{\star}\right\Vert \nonumber \\
 & \leq & 2\rho^{\left(n-p\right)\wedge p}\left(\frac{\epsilon^{+}}{\epsilon^{-}}\right)^{2}\left(1+2\left(\frac{\epsilon^{+}}{\epsilon^{-}}\right)^{2}\right)\label{eq:lambda_converge}
\end{eqnarray}
where the penultimate inequality is due to (\ref{eq:MET_bound}) of
Theorem \ref{thm:MET} and (\ref{eq:eta_converge}), and the final
inequality is due to (\ref{eq:eta_in_P_h_bound}). Integrating and
iterating the eigen-measure equation (\ref{eq:eigen_star_defns})
gives $\lambda_{\star}^{n}=\eta_{\star}Q^{n}(1)$ . Then by Lemma
\ref{lem:Q_n_bounds},
\begin{equation}
\sup_{n\geq1}\sup_{x\in\mathsf{X}}\frac{Q^{\left(n\right)}(1)(x)}{\lambda_{\star}^{n}}\leq\frac{\epsilon^{+}}{\epsilon^{-}}.\label{eq:Qt_unif}
\end{equation}
With the above bounds in hand we now address (\ref{eq:h_converge}).
We have 
\begin{eqnarray*}
\left|h_{p,n}(x)-h_{\star}(x)\right| & = & \left|\frac{Q^{\left(n-p\right)}(1)(x)}{\lambda_{\star}^{n-p}}\left(\frac{\lambda_{\star}^{n-p}}{\prod_{\ell=p}^{n-1}\lambda_{\ell}}-1\right)+\frac{Q^{\left(n-p\right)}(1)(x)}{\lambda_{\star}^{n-p}}-h_{\star}(x)\right|\\
 & \leq & \left|\frac{\lambda_{\star}^{n-p}}{\prod_{\ell=p}^{n-1}\lambda_{\ell}}-1\right|\sup_{m\geq1}\sup_{y\in\mathsf{X}}\frac{Q^{\left(m\right)}(1)(y)}{\lambda_{\star}^{m}}\\
 &  & +\left|\frac{Q^{\left(n-p\right)}(1)(x)}{\lambda_{\star}^{n-p}}-h_{\star}(x)\right|\\
 & \leq & 2\rho^{\left(n-p\right)\wedge p}\left(\frac{\epsilon^{+}}{\epsilon^{-}}\right)^{3}\left(1+2\left(\frac{\epsilon^{+}}{\epsilon^{-}}\right)^{2}\right)+2\rho^{n-p}\left(\frac{\epsilon^{+}}{\epsilon^{-}}\right)^{2}\\
 & = & 2\rho^{\left(n-p\right)\wedge p}\left(\frac{\epsilon^{+}}{\epsilon^{-}}\right)^{2}\left[1+\left(\frac{\epsilon^{+}}{\epsilon^{-}}\right)+2\left(\frac{\epsilon^{+}}{\epsilon^{-}}\right)^{3}\right].
\end{eqnarray*}
where for the final inequality, (\ref{eq:lambda_converge}), (\ref{eq:Qt_unif})
and (\ref{eq:MET_bound}) have been used. This establishes (\ref{eq:h_converge}).

For (\ref{eq:P_converge}), consider the decomposition
\begin{eqnarray*}
\interleave P_{(p,n)}-P_{\star}\interleave & \leq & \sup_{x}\sup_{\varphi:\left|\varphi\right|\leq1}\left[\frac{1}{\lambda_{p-1}h_{p-1,n}(x)}\left|Q\left[\left(h_{p,n}-h_{\star}\right)\varphi\right](x)\right|\right.\\
 &  & +\frac{1}{\lambda_{p-1}}\frac{\left|h_{p-1,n}(x)-h_{\star}(x)\right|}{h_{p-1,n}(x)}\frac{\left|Q\left(h_{\star}\varphi\right)(x)\right|}{h_{\star}(x)}\\
 &  & +\left.\frac{\left|\lambda_{\star}-\lambda_{p-1}\right|}{\lambda_{p-1}\lambda_{\star}}\frac{1}{h_{\star}(x)}\left|Q\left(h_{\star}\varphi\right)(x)\right|\right]\\
 & \leq & \left\Vert h_{p,n}-h_{\star}\right\Vert \sup_{x}\frac{Q(1)(x)}{\lambda_{p-1}h_{p-1,n}(x)}\\
 &  & +\frac{\lambda_{\star}}{\lambda_{p-1}}\left\Vert h_{p-1,n}-h_{\star}\right\Vert \sup_{x}\frac{1}{h_{p-1,n}(x)}\\
 &  & +\frac{\left|\lambda_{\star}-\lambda_{p-1}\right|}{\lambda_{p-1}}\\
 & \leq & C_{h}\rho^{\left(n-p\right)\wedge p}2\left(\frac{\epsilon^{+}}{\epsilon^{-}}\right)^{2}+C_{\eta}\rho^{p-1}\frac{\epsilon^{+}}{\epsilon^{-}},
\end{eqnarray*}
where for the final equality, Lemma \ref{lem:h_p_n_bounds}, the identities
$\lambda_{p}=\eta_{p}(G)$, $\lambda_{\star}=\eta_{\star}(G)$, and
(\ref{eq:eta_converge})-(\ref{eq:h_converge}) have been used.
\end{proof}

\subsection{Proofs and auxiliary results for Section \ref{sec:Particle-approximations}}

\subsubsection{Lack of bias}
\begin{proof}
(of Proposition \ref{prop:unbias}). The $n=0$ case is trivial. For
any $\varphi\in\mathcal{L},$ $n\geq1$ and $x\in\mathsf{X},$ we
have
\begin{eqnarray}
\mathbb{E}_{N}\left[\left.Q_{n}^{N}\left(\varphi\right)(x)\right|\mathcal{F}_{n-1}\right] & = & \mathbb{E}_{N}\left[\left.\int\frac{\mathrm{d}Q\left(x,\cdot\right)}{\mathrm{d}\Phi\left(\eta_{n-1}^{N}\right)}\left(x^{\prime}\right)\varphi\left(x^{\prime}\right)\eta_{n}^{N}\left(dx^{\prime}\right)\right|\mathcal{F}_{n-1}\right]\nonumber \\
 & = & \frac{1}{N}\sum_{i=1}^{N}\mathbb{E}_{N}\left[\left.\frac{\mathrm{d}Q\left(x,\cdot\right)}{\mathrm{d}\Phi\left(\eta_{n-1}^{N}\right)}\left(\zeta_{n}^{i}\right)\varphi\left(\zeta_{n}^{i}\right)\right|\mathcal{F}_{n-1}\right]\nonumber \\
 & = & \int\frac{\mathrm{d}Q\left(x,\cdot\right)}{\mathrm{d}\Phi\left(\eta_{n-1}^{N}\right)}\left(x^{\prime}\right)\varphi\left(x^{\prime}\right)\Phi\left(\eta_{n-1}^{N}\right)\left(dx^{\prime}\right)\nonumber \\
 & = & Q\left(\varphi\right)(x),\label{eq:unbiased_one_step}
\end{eqnarray}
where the penultimate equality is due to the definition of the particle
transition probabilities (\ref{eq:particle transitions}).

Now consider the telescoping decomposition
\begin{eqnarray*}
\mu^{N}Q_{0,n}^{N}\left(\varphi\right)-\mu^{\prime}Q^{\left(n\right)}\left(\varphi\right) & = & \sum_{p=0}^{n-1}\left[\mu^{N}Q_{0,p+1}^{N}Q^{\left(n-p-1\right)}\left(\varphi\right)-\mu^{N}Q_{0,p}^{N}Q^{\left(n-p\right)}\left(\varphi\right)\right]\\
 &  & +\left(\mu^{N}-\mu^{\prime}\right)Q^{\left(n\right)}\left(\varphi\right).
\end{eqnarray*}
For each term in the big summation we have
\begin{eqnarray*}
 &  & \mathbb{E}_{N}\left[\left.\mu^{N}Q_{0,p+1}^{N}Q^{\left(n-p-1\right)}\left(\varphi\right)-\mu^{N}Q_{0,p}^{N}Q^{\left(n-p\right)}\left(\varphi\right)\right|\mathcal{F}_{p}\right]\\
 &  & =\int\mu^{N}Q_{0,p}^{N}\left(dx_{p}\right)\mathbb{E}\left[\left.Q_{p+1}^{N}Q^{\left(n-p-1\right)}\left(\varphi\right)(x_{p})-Q^{\left(n-p\right)}\left(\varphi\right)(x_{p})\right|\mathcal{F}_{p}\right]\\
 &  & =0,
\end{eqnarray*}
where the final equality is due to (\ref{eq:unbiased_one_step}).
For the remaining term, $\mathbb{E}_{N}\left[\left(\mu^{N}-\mu^{\prime}\right)Q^{\left(n\right)}\left(\varphi\right)\right]=0$
by assumption of the proposition.
\end{proof}

\subsubsection{Path-wise stability}

The following proposition provides a generic result on iterates of
non-negative kernels, which will serve multiple purposes throughout
the remaining proofs in the paper. 
\begin{prop}
\label{prop_generic_stab}Let $\left(K_{n};n\geq1\right)$ be a collection
of possibly random, non-negative integral kernels, and suppose that
for a collection of possibly random, finite measures $\left(\nu_{n};n\geq1\right)$
and positive, bounded functions $\left(S_{n}^{-},S_{n}^{+};n\geq1\right)$,
\begin{equation}
S_{n}^{-}(x)\nu_{n}\left(\cdot\right)\leq K_{n}(x,\cdot)\leq S_{n}^{+}(x)\nu_{n}\left(\cdot\right),\quad\forall x\in\mathsf{X},n\geq1.\label{eq:generic_mixing}
\end{equation}
Then 
\begin{equation}
\sup_{n\geq1}\sup_{x,x^{\prime}\in\mathsf{X}}\frac{K_{0,n}(1)(x)}{K_{0,n}(1)(x^{\prime})}\leq\sup_{n\geq1}\overline{S}_{n},\label{eq:K_ratio_bounded}
\end{equation}
where
\[
\overline{S}_{n}:=\sup_{x,x^{\prime}\in\mathsf{X}}\frac{S_{n}^{+}(x)}{S_{n}^{-}(x^{\prime})}.
\]
Furthermore, for any possibly random probability measure $\eta$ and
$\varphi\in\mathcal{L},$ 
\[
\sup_{x\in\mathsf{X}}\left|\frac{K_{0,n}(\varphi)(x)}{\eta K_{0,n}(1)}-\frac{K_{0,n}(1)(x)}{\eta K_{0,n}(1)}\frac{\eta K_{0,n}(\varphi)}{\eta K_{0,n}(1)}\right|\leq\left\Vert \varphi\right\Vert 2C_{S}\prod_{p=1}^{n}\rho_{p}
\]
 where $\rho_{n}:=1-\left(\inf_{x\in\mathsf{X}}\frac{S{}_{n}^{-}(x)}{S_{n}^{+}(x)}\right)^{2}$
and $C_{S}:=\sup_{n\geq1}\overline{S}_{n}.$\end{prop}
\begin{rem}
We approach the proof of this proposition using a decomposition idea
of \citet{filter:the:KV08}, a technique which they demonstrated to
be useful in the analysis of non-linear filter stability on non-compact
state-spaces. We won't exploit the full generality of this kind of
decomposition (it is useful under conditions much weaker than \textbf{(H)}
- see for example \citep{filter:the:DFMP09}, again in the filtering
context) and we choose to take this approach because it yields a short
and direct proof, which is sufficient for our purposes. \end{rem}
\begin{proof}
(of Proposition \ref{prop_generic_stab}). The uniform bound of (\ref{eq:K_ratio_bounded})
holds directly under the assumptions of the proposition.

We write $K_{n}^{\otimes2}\left(x,y,d\left(x^{\prime},y^{\prime}\right)\right):=K_{n}\left(x,dx^{\prime}\right)K_{n}\left(y,dy^{\prime}\right)$
and $\nu_{n}^{\otimes2}\left(d\left(x,y\right)\right):=\nu_{n}(dx)\nu_{n}(dy)$.
Under the assumptions of the proposition we have for any $\left(x,y\right)\in\mathsf{X}^{2}$
and measurable $A\subset\mathsf{X}^{2}$ such that $\nu_{n}^{\otimes2}(A)>0$,
\begin{eqnarray}
\widehat{K}_{n}\left(x,y,A\right) & := & K_{n}^{\otimes2}\left(x,y,A\right)-S{}_{n}^{-}(x)S{}_{n}^{-}(y)\nu_{n}^{\otimes2}\left(A\right)\nonumber \\
 & \leq & \left[1-\frac{S{}_{n}^{-}(x)S{}_{n}^{-}(y)}{S_{n}^{+}(x)S_{n}^{+}(y)}\right]K_{n}^{\otimes2}\left(x,y,A\right).\nonumber \\
 & \leq & \rho_{n}K_{n}^{\otimes2}\left(x,y,A\right).\label{eq:K_hat_bound}
\end{eqnarray}
 Furthermore, 
\begin{eqnarray}
 &  & \left|\frac{K_{0,n}(\varphi)(x)}{\eta K_{0,n}(1)}-\frac{K_{0,n}(1)(x)}{\eta K_{0,n}(1)}\frac{\eta K_{0,n}(\varphi)}{\eta K_{0,n}(1)}\right|\nonumber \\
 &  & =\frac{\left|K_{0,n}(\varphi)(x)\eta K_{0,n}(1)-K_{0,n}(1)(x)\eta K_{0,n}(\varphi)\right|}{\eta K_{0,n}(1)\eta K_{0,n}(1)}\nonumber \\
 &  & =\frac{K_{0,n}(1)(x)}{\eta K_{0,n}(1)}\frac{\left|\left(\delta_{x}\otimes\eta\right)K_{0,n}^{\otimes2}\left(\varphi\otimes1-1\otimes\varphi\right)\right|}{\left(\delta_{x}\otimes\eta\right)K_{0,n}^{\otimes2}(1\otimes1)}\nonumber \\
 &  & =\frac{K_{0,n}(1)(x)}{\eta K_{0,n}(1)}\frac{\left|\left(\delta_{x}\otimes\eta\right)\widehat{K}_{0,n}\left(\varphi\otimes1-1\otimes\varphi\right)\right|}{\left(\delta_{x}\otimes\eta\right)K_{0,n}^{\otimes2}(1\otimes1)}\label{eq:veret}\\
 &  & \leq2\left\Vert \varphi\right\Vert \left(\sup_{p\geq1}\overline{S}_{p}\right)\frac{\left(\delta_{x}\otimes\eta\right)\widehat{K}_{0,n}\left(1\otimes1\right)}{\left(\delta_{x}\otimes\eta\right)K_{0,n}^{\otimes2}(1\otimes1)}\nonumber \\
 &  & \leq2\left\Vert \varphi\right\Vert \left(\sup_{p\geq1}\overline{S}_{p}\right)\prod_{p=1}^{n}\rho_{p},\nonumber 
\end{eqnarray}
where the equality in (\ref{eq:veret}) is due to the decomposition
technique of \citet[p. 422]{filter:the:KV08} \citep[see also][Proof of Proposition 12]{filter:the:DFMP09},
and for the final two inequalities (\ref{eq:K_ratio_bounded}) and
(\ref{eq:K_hat_bound}) have been used. 
\end{proof}
Under assumption \textbf{(H)}, we find that the random operators satisfy
path-wise, a regularity condition of a similar form, which is used
below in the Proof of Proposition \ref{prop:intermediate_L_r}.
\begin{lem}
\label{lem:randomized_mixing}The operators $\left(Q_{n}^{N}\right)$
satisfy 
\begin{equation}
\alpha_{n}^{N}(\cdot)\epsilon^{-}\leq Q_{n}^{N}(x,\cdot)\leq\epsilon^{+}\alpha_{n}^{N}(\cdot),\quad\forall x\in\mathsf{X},n\geq1,N\geq1,\label{eq:Q_^N_mixing}
\end{equation}
where $\alpha_{n}^{N}$ is the random finite measure: 
\[
\alpha_{n}^{N}\left(dx\right):=\eta_{n}^{N}(dx)\left[\frac{\mathrm{d}\Phi\left(\eta_{n-1}^{N}\right)}{\mathrm{d}\nu}(x)\right]^{-1},
\]
and $\epsilon^{-},\epsilon^{+}$ are the deterministic constants in
assumption \textbf{(H)}. Moreover for all $x\in\mathsf{X}$ and $p\leq n$,
\[
\frac{\epsilon^{-}}{\epsilon^{+}}\leq h_{p,n}^{N}(x)\leq\frac{\epsilon^{+}}{\epsilon^{-}},
\]
\end{lem}
\begin{proof}
Since $Q(x,\cdot)$ is equivalent to $\nu$, then $\Phi\left(\eta_{n-1}^{N}\right)$
is too, and it is straightforward to check that assumption \textbf{(H)}
implies that $\frac{\mathrm{d}\nu}{\mathrm{d}\Phi\left(\eta_{n-1}^{N}\right)}(x)$
is bounded above and below away from zero in $x$. We then have 
\begin{eqnarray*}
Q_{n}^{N}(x,A) & = & \int_{A}\frac{\mathrm{d}Q(x\cdot)}{\mathrm{d}\Phi\left(\eta_{n-1}^{N}\right)}(x^{\prime})\eta_{n}^{N}(dx^{\prime})\\
 & = & \int_{A}q(x,x^{\prime})\frac{\mathrm{d}\nu}{\mathrm{d}\Phi\left(\eta_{n-1}^{N}\right)}(x^{\prime})\eta_{n}^{N}(dx^{\prime})\\
 & \leq & \epsilon^{+}\int_{A}\frac{\mathrm{d}\nu}{\mathrm{d}\Phi\left(\eta_{n-1}^{N}\right)}(x^{\prime})\eta_{n}^{N}(dx^{\prime}),
\end{eqnarray*}
The proof of the lower bound is similar. The bounds for $h_{p,n}^{N}(x)=Q_{p,n}^{N}(1)(x)/\eta_{p}^{N}Q_{p,n}^{N}(1)$
follow from (\ref{eq:Q_^N_mixing}).
\end{proof}

\begin{proof}
(of Theorem \ref{thm_pathwise}) From Lemma \ref{lem:_random_eigen_equations},
\begin{equation}
\prod_{p=0}^{n-1}\lambda_{p}^{N}=\eta_{0}^{N}Q_{0,n}^{N}\left(1\right),\quad\quad h_{0,n}^{N}=\frac{Q_{0,n}^{N}\left(1\right)}{\eta_{0}^{N}Q_{0,n}^{N}\left(1\right)},\quad\quad\eta_{n}^{N}=\frac{\eta_{0}^{N}Q_{0,n}^{N}}{\eta_{0}^{N}Q_{0,n}^{N}\left(1\right)}.\label{eq:eigen_relations_reminder}
\end{equation}
Thus (\ref{eq:particle_path_met}) holds due to Lemma \ref{lem:randomized_mixing}
and Proposition \ref{prop_generic_stab} applied with $\eta=\eta_{0}^{N}$,
$K_{n}=Q_{n}^{N},$ $\nu_{n}=\alpha_{n}^{N}$ and $S_{n}^{+}=\epsilon^{+},S_{n}^{-}=\epsilon^{-}$
are constant. Dividing through by $\mu^{\prime}\left(h_{0.n}^{N}\right)$
in (\ref{eq:particle_path_met}), again noting (\ref{eq:eigen_relations_reminder})
and using 
\begin{equation}
\sup_{n\ge1}\sup_{x,x^{\prime}\in\mathsf{X}}\frac{Q_{0,n}^{N}(1)(x)}{Q_{0,n}^{N}(1)(x^{\prime})}\leq\frac{\epsilon^{+}}{\epsilon^{-}},\label{eq:QN_ratio_bound}
\end{equation}
which also holds by Proposition \ref{prop_generic_stab}, we establish
(\ref{eq:particle_path_filter_stab}).
\end{proof}

\subsection{Auxiliary results and proof of Theorem \ref{thm:L_r_bounds}}

Consider the collection of ``backward'' random kernels $\left(R_{n}^{N}\right)$
defined by
\begin{eqnarray*}
R_{n}^{N}\left(x,dx^{\prime}\right) & := & \eta_{n-1}^{N}(dx^{\prime})\frac{dQ(x^{\prime},\cdot)}{d\Phi\left(\eta_{n-1}^{N}\right)}(x),\quad n\geq1,
\end{eqnarray*}
and with a slight abuse of convention, write
\[
R_{n,n}^{N}:=Id,\quad\quad R_{n,p}^{N}:=R_{n}^{N}R_{n-1}^{N}\cdots R_{p+1}^{N},\quad p<n.
\]

The interest in these quantities is that, in the context of the $L_{r}$
error estimates which are the focus of this section, they provide
a convenient way to express the functions $\left(h_{p,n}^{N}\right)$
and share path-wise stability properties with $\left(Q_{n}^{N}\right)$.
Indeed by a simple induction it can be shown that for any $\varphi\in\mathcal{L}$,
\begin{equation}
\eta_{n}^{N}R_{n,p}^{N}\left(\varphi\right)=\eta_{p}^{N}\left[\varphi Q_{p,n}^{N}(1)\right],\quad p\leq n.\label{eq:eta_R=00003Deta_Q}
\end{equation}

\begin{rem}
Each kernel $R_{n}^{N}$ is equal, up to a scaling factor of $\eta_{n-1}^{N}(G)$,
to a certain ``backward'' Markov kernel used in the analysis of
\citet{del2010backward}. In contrast to the latter work, we are centrally
concerned with emphasizing the relationship between $\left(Q_{p,n}^{N}\right)$
and the underlying semigroup $\left(Q^{(n)}\right)$. In view of (\ref{eq:eta_R=00003Deta_Q})
and Proposition \ref{prop:unbias}, we therefore prefer to deal with
$\left(R_{n}^{N}\right)$, but only for cosmetic reasons.

The $\left(R_{n}^{N}\right)$ satisfy a condition similar to that
in Lemma \ref{lem:randomized_mixing}, as per the following Lemma.\end{rem}
\begin{lem}
\label{lem:randomized_mixing-1}The operators $\left(R_{n}^{N}\right)$
satisfy 
\[
\eta_{n-1}^{N}(\cdot)\beta_{n}^{N}(x)\epsilon^{-}\leq R_{n}^{N}(x,\cdot)\leq\epsilon^{+}\beta_{n}^{N}(x)\eta_{n-1}^{N}(\cdot),\quad\forall x\in\mathsf{X},n\geq1,N\geq1,
\]
 where $\beta_{n}^{N}$ is the random, positive and bounded function:
\[
\beta_{n}^{N}\left(x\right):=\left[\frac{\mathrm{d}\Phi\left(\eta_{n-1}^{N}\right)}{\mathrm{d}\nu}(x)\right]^{-1},
\]
and $\epsilon^{-},\epsilon^{+}$ are the deterministic constants in
assumption \textbf{(H)}.\end{lem}
\begin{proof}
From definitions,
\begin{eqnarray*}
R_{n}^{N}(x,A) & = & \int_{A}\frac{\mathrm{d}Q(x^{\prime},\cdot)}{\mathrm{d}\Phi\left(\eta_{n-1}^{N}\right)}(x)\eta_{n-1}^{N}(dx^{\prime})\\
 & = & \int_{A}\frac{\mathrm{d}Q(x^{\prime},\cdot)}{\mathrm{d}\nu}\frac{d\nu}{\Phi\left(\eta_{n-1}^{N}\right)}(x)\eta_{n-1}^{N}(dx^{\prime})\\
 & \leq & \epsilon^{+}\frac{\mathrm{d}\nu}{\mathrm{d}\Phi\left(\eta_{n-1}^{N}\right)}(x)\eta_{n-1}^{N}\left(A\right).
\end{eqnarray*}
The claimed positivity and boundedness of $\beta_{n}^{N}$ follows
from \textbf{(H)}. The proof of the lower bound is similar.
\end{proof}
It is well known that under \textbf{(H)} and variations thereof, one
can obtain time-uniform $L_{r}$ estimates for errors of the form
$\eta_{n}^{N}(\varphi)-\eta_{n}\left(\varphi\right)$. We will make
use of the following result, due to \citet[Theorem 7.4.4]{smc:theory:Dm04}.
The proof is omitted.
\begin{prop}
\label{prop_eta_uniform}For any $r\geq1$ there exists a universal
constant $B_{r}$ such that for any $\varphi\in\mathcal{L},$ the
following time uniform estimate holds
\[
\sup_{n\geq0}\mathbb{E}_{N}\left[\left|\eta_{n}^{N}\left(\varphi\right)-\eta_{n}\left(\varphi\right)\right|^{r}\right]^{1/r}\leq2\left\Vert \varphi\right\Vert \frac{B_{r}}{\sqrt{N}}\left(\frac{\epsilon^{+}}{\epsilon^{-}}\right)^{5}.
\]

\end{prop}
We need a further definition. Consider now the functions $\left(\phi_{n}\right)$
and their random counterparts $\left(\phi_{n}^{N}\right)$ defined
by 
\[
\phi_{n}\left(x,x^{\prime}\right):=\frac{dQ\left(x,\cdot\right)}{d\eta_{n}Q}\left(x^{\prime}\right),\quad\quad\phi_{n}^{N}\left(x,x^{\prime}\right):=\frac{dQ\left(x,\cdot\right)}{d\eta_{n}^{N}Q}\left(x^{\prime}\right),\quad n\geq0
\]
and note that under \textbf{(H)}, 
\begin{equation}
\sup_{n\geq0}\sup_{x,x^{\prime}}\left|\phi_{n}\left(x,x^{\prime}\right)\right|\leq\frac{\epsilon^{+}}{\epsilon^{-}},\quad\quad\sup_{N\geq1}\sup_{n\geq0}\sup_{x,x^{\prime}}\left|\phi_{n}^{N}\left(x,x^{\prime}\right)\right|\leq\frac{\epsilon^{+}}{\epsilon^{-}}.\label{eq:phi_bound}
\end{equation}
Furthermore, we then have from definitions that
\begin{eqnarray}
h_{p,n}^{N}(x) & = & \frac{Q_{p,n}^{N}(1)(x)}{\eta_{p}^{N}Q_{p,n}^{N}(1)}=\frac{1}{\eta_{p+1}^{N}Q_{p+1,n}^{N}(1)}\int\frac{dQ\left(x,\cdot\right)}{d\eta_{p}^{N}Q}\left(x^{\prime}\right)Q_{p+1,n}^{N}(1)(x^{\prime})\eta_{p+1}^{N}\left(dx^{\prime}\right)\nonumber \\
 &  & =\frac{\eta_{n}^{N}R_{n,p+1}^{N}\left[\phi_{p}^{N}\left(x,\cdot\right)\right]}{\eta_{n}^{N}R_{n,p+1}^{N}\left(1\right)},\label{eq:h_Q_R_id}
\end{eqnarray}
where the final equality is due to (\ref{eq:eta_R=00003Deta_Q}).
\begin{prop}
\label{prop:intermediate_L_r}For any $r\geq1$ there exists a universal
constant $B_{r}$ such that for any $\varphi\in\mathcal{L}$ and $N\geq1$,
\[
\sup_{p\leq n}\sup_{x\in\mathsf{X}}\mathbb{E}_{N}\left[\left|\frac{Q_{p}^{N}(\varphi h_{p,n}^{N})(x)}{\lambda_{p-1}^{N}}-\frac{Q(\varphi h_{p,n})(x)}{\lambda_{p-1}}\right|^{r}\right]^{1/r}\leq2\left\Vert \varphi\right\Vert \frac{B_{r}}{\sqrt{N}}\tilde{C},
\]
where
\[
\tilde{C}=\left[3\left(\frac{\epsilon^{+}}{\epsilon^{-}}\right)^{7}+\left(\frac{\epsilon^{+}}{\epsilon^{-}}\right)^{5}\frac{1}{1-\tilde{\rho}}\right],
\]
and $\tilde{\rho}$ is as in Theorem \ref{thm_pathwise}.\end{prop}
\begin{proof}
(of Proposition \ref{prop:intermediate_L_r}) From the identities
\[
\frac{Q_{p}^{N}(\varphi h_{p,n}^{N})(x)}{\lambda_{p-1}^{N}}=\frac{\eta_{p}^{N}\left[\varphi\phi_{p-1}^{N}\left(x,\cdot\right)Q_{p,n}^{N}(1)\right]}{\eta_{p}^{N}Q_{p,n}^{N}(1)}=\frac{\eta_{n}^{N}R_{n,p}^{N}\left[\varphi\phi_{p-1}^{N}\left(x,\cdot\right)\right]}{\eta_{n}^{N}R_{n,p}^{N}\left(1\right)}
\]
(established similarly to equation (\ref{eq:h_Q_R_id})) and 
\[
\frac{Q(\varphi h_{p,n})}{\lambda_{p-1}}(x)=\frac{\eta_{p}\left[\varphi\phi_{p-1}\left(x,\cdot\right)Q^{(n-p)}(1)\right]}{\eta_{p}Q^{\left(n-p\right)}(1)},
\]
we have the decomposition
\[
\frac{Q_{p}^{N}(\varphi h_{p,n}^{N})(x)}{\lambda_{p-1}^{N}}-\frac{Q(\varphi h_{p,n})(x)}{\lambda_{p-1}}=\sum_{j=1}^{3}T_{p,n}^{N,j}(x)
\]
where
\begin{eqnarray}
T_{p,n}^{N,1}(x) & := & \frac{\eta_{n}^{N}R_{n,p}^{N}\left[\varphi\left(\phi_{p-1}^{N}\left(x,\cdot\right)-\phi_{p-1}\left(x,\cdot\right)\right)\right]}{\eta_{n}^{N}R_{n,p}^{N}\left(1\right)}\label{eq:h_decomp1-1}\\
T_{p,n}^{N,2}(x) & := & \frac{\eta_{n}^{N}R_{n,p}^{N}\left[\varphi\phi_{p-1}\left(x,\cdot\right)\right]}{\eta_{n}^{N}R_{n,p}^{N}\left(1\right)}-\frac{\Phi\left(\eta_{p-1}^{N}\right)\left[\varphi\phi_{p-1}\left(x,\cdot\right)Q^{(n-p)}(1)\right]}{\Phi\left(\eta_{p-1}^{N}\right)Q^{\left(n-p\right)}(1)}\label{eq:h_decomp2-1}\\
T_{p,n}^{N,3}(x) & := & \frac{\Phi\left(\eta_{p-1}^{N}\right)\left[\varphi\phi_{p-1}\left(x,\cdot\right)Q^{(n-p)}(1)\right]}{\Phi\left(\eta_{p-1}^{N}\right)Q^{\left(n-p\right)}(1)}-\frac{\eta_{p}\left[\varphi\phi_{p-1}\left(x,\cdot\right)Q^{(n-p)}(1)\right]}{\eta_{p}Q^{\left(n-p\right)}(1)}.\label{eq:h_decomp3-1}
\end{eqnarray}
For the difference in (\ref{eq:h_decomp1-1}), under \textbf{(H)}
we have
\begin{eqnarray*}
\sup_{x\in\mathsf{X}}\left|T_{p,n}^{N,1}(x)\right| & \leq & \frac{\left\Vert \varphi\right\Vert \epsilon^{+}}{\eta_{n}^{N}R_{n,p}^{N}\left(1\right)}\int\left|\frac{1}{\int\eta_{p-1}^{N}(dy)q\left(y,x^{\prime}\right)}-\frac{1}{\int\eta_{p-1}(dy)q\left(y,x^{\prime}\right)}\right|\eta_{n}^{N}R_{n,p}^{N}\left(dx^{\prime}\right)\\
 & \leq & \frac{\left\Vert \varphi\right\Vert \epsilon^{+}}{\eta_{n}^{N}R_{n,p}^{N}\left(1\right)}\int\left|\frac{\int q\left(y,x^{\prime}\right)\left[\eta_{p-1}(dy)-\eta_{p-1}^{N}(dy)\right]}{\int q\left(y,x^{\prime}\right)\eta_{p-1}^{N}(dy)\int q\left(y,x^{\prime}\right)\eta_{p-1}(dy)}\right|\eta_{n}^{N}R_{n,p}^{N}\left(dx^{\prime}\right)\\
 & \leq & \frac{\left\Vert \varphi\right\Vert }{\left(\epsilon^{-}\right)^{2}}\frac{\epsilon^{+}}{\eta_{n}^{N}R_{n,p}^{N}\left(1\right)}\int\left|\int q\left(y,x^{\prime}\right)\left[\eta_{p-1}(dy)-\eta_{p-1}^{N}(dy)\right]\right|\eta_{n}^{N}R_{n,p}^{N}\left(dx^{\prime}\right)\\
 & \leq & \left\Vert \varphi\right\Vert \frac{\epsilon^{+}}{\left(\epsilon^{-}\right)^{2}}\sup_{x^{\prime}}\left|\int q\left(y,x^{\prime}\right)\left[\eta_{p-1}(dy)-\eta_{p-1}^{N}(dy)\right]\right|,
\end{eqnarray*}
and therefore by Proposition \ref{prop_eta_uniform} and $q\left(y,x^{\prime}\right)\leq\epsilon^{+}$,
\begin{equation}
\sup_{x\in\mathsf{X}}\mathbb{E}_{N}\left[\left|T_{p,n}^{N,1}(x)\right|^{r}\right]^{1/r}\leq2\left\Vert \varphi\right\Vert \frac{B_{r}}{\sqrt{N}}\left(\frac{\epsilon^{+}}{\epsilon^{-}}\right)^{7}.\label{eq:T_1_final-1}
\end{equation}
For the difference in (\ref{eq:h_decomp2-1}), due to the relation
\[
\eta_{p-1}^{N}(dx)Q\left(x,dx^{\prime}\right)=\Phi\left(\eta_{p-1}^{N}\right)\left(dx^{\prime}\right)R_{p}^{N}\left(x^{\prime},dx\right),
\]
we have the telescoping decomposition
\begin{eqnarray}
 &  & T_{p,n}^{N,2}(x)\nonumber \\
 &  & =\frac{\eta_{n}^{N}R_{n,p}^{N}\left[\varphi\phi_{p-1}\left(x,\cdot\right)\right]}{\eta_{n}^{N}R_{n,p}^{N}\left(1\right)}-\frac{\Phi\left(\eta_{p-1}^{N}\right)\left[\varphi\phi_{p-1}\left(x,\cdot\right)Q^{(n-p)}(1)\right]}{\Phi\left(\eta_{p-1}^{N}\right)Q^{\left(n-p\right)}(1)}\nonumber \\
 &  & =\sum_{m=p}^{n}\left[\frac{\eta_{m}^{N}\left[Q^{\left(n-m\right)}(1)R_{m,p}^{N}\left[\varphi\phi_{p-1}\left(x,\cdot\right)\right]\right]}{\eta_{m}^{N}\left[Q^{\left(n-m\right)}(1)R_{m,p}^{N}(1)\right]}-\frac{\Phi\left(\eta_{m-1}^{N}\right)\left[Q^{\left(n-m\right)}(1)R_{m,p}^{N}\left[\varphi\phi_{p-1}\left(x,\cdot\right)\right]\right]}{\Phi\left(\eta_{m-1}^{N}\right)\left[Q^{\left(n-m\right)}(1)R_{m,p}^{N}(1)\right]}\right].\label{eq:summation-1}
\end{eqnarray}
Each term in the summation (\ref{eq:summation-1}) is of the form
\begin{equation}
\frac{\Phi\left(\eta_{m-1}^{N}\right)\left[Q^{\left(n-m\right)}R_{m,p}^{N}\left(1\right)\right]}{\eta_{m}^{N}\left[Q^{\left(n-m\right)}R_{m,p}^{N}\left(1\right)\right]}\left[\eta_{m}^{N}-\Phi\left(\eta_{m-1}^{N}\right)\right]\left[\Delta_{p,n,m}^{(x)}\right],\label{eq:one_term-1}
\end{equation}
where 
\begin{eqnarray*}
\Delta_{p,n,m}^{(x)}(y) & := & \frac{Q^{\left(n-m\right)}(1)(y)R_{m,p}^{N}\left[\varphi\phi_{p}\left(x,\cdot\right)\right](y)}{\Phi\left(\eta_{m-1}^{N}\right)\left[Q^{\left(n-m\right)}(1)R_{m,p}^{N}\left(1\right)\right]}\\
 &  & -\frac{Q^{\left(n-m\right)}(1)(y)R_{m,p}^{N}\left(1\right)(y)}{\Phi\left(\eta_{m-1}^{N}\right)\left[Q^{\left(n-m\right)}(1)R_{m,p}^{N}\left(1\right)\right]}\frac{\Phi\left(\eta_{m-1}^{N}\right)\left[Q^{\left(n-m\right)}(1)R_{m,p}^{N}\left[\varphi\phi_{p-1}\left(x,\cdot\right)\right]\right]}{\Phi\left(\eta_{m-1}^{N}\right)\left[Q^{\left(n-m\right)}(1)R_{m,p}^{N}\left(1\right)\right]}.
\end{eqnarray*}
Defining the map $\Psi_{m,n}:\mathcal{P}\rightarrow\mathcal{P}$ by
$\Psi_{m,n}(\eta)(A):=\dfrac{\eta\left[Q^{\left(n-m\right)}(1)\mathbb{I}_{A}\right]}{\eta Q^{\left(n-m\right)}(1)}$,
for $A\in\mathcal{B}$, we have

\begin{eqnarray*}
 &  & \sup_{x,y}\left|\Delta_{p,n,m}^{(x)}(y)\right|\\
 & \leq & \sup_{y}\frac{Q^{\left(n-m\right)}(1)(y)}{\Phi\left(\eta_{m-1}^{N}\right)\left[Q^{\left(n-m\right)}(1)\right]}\\
 &  & \times\sup_{x,y}\left|\frac{R_{m,p}^{N}\left[\varphi\phi_{p-1}\left(x,\cdot\right)\right](y)}{\Psi_{m,n}\left[\Phi\left(\eta_{m-1}^{N}\right)\right]\left[R_{m,p}^{N}\left(1\right)\right]}-\frac{R_{m,p}^{N}\left(1\right)(y)}{\Psi_{m,n}\left[\Phi\left(\eta_{m-1}^{N}\right)\right]\left[R_{m,p}^{N}\left(1\right)\right]}\frac{\Psi_{m,n}\left[\Phi\left(\eta_{m-1}^{N}\right)\right]R_{m,p}^{N}\left[\varphi\phi_{p-1}\left(x,\cdot\right)\right]}{\Psi_{m,n}\left[\Phi\left(\eta_{m-1}^{N}\right)\right]\left[R_{m,p}^{N}\left(1\right)\right]}\right|\\
 & \leq & \left\Vert \varphi\right\Vert \tilde{\rho}^{m-p}2\left(\frac{\epsilon^{+}}{\epsilon^{-}}\right)^{3}.
\end{eqnarray*}
where the inequality is due to Lemma \ref{lem:Q_n_bounds}, the bound
of (\ref{eq:phi_bound}) and then Lemma \ref{lem:randomized_mixing-1}
and Proposition \ref{prop_generic_stab} applied to the sequence of
kernels $R_{m}^{N},R_{m-1}^{N},\ldots,R_{p+1}^{N}$ with $\eta=\Psi_{m,n}\left[\Phi\left(\eta_{m-1}^{N}\right)\right]$,
and $\tilde{\rho}$ is as in Theorem \ref{thm_pathwise}. Then returning
to (\ref{eq:summation-1})-(\ref{eq:one_term-1}), and noting that
$\Delta_{p,n,m}^{(x)}(y)$ is measurable w.r.t. to $\mathcal{F}_{m-1}$,
we have by an application of \citet[Lemma 7.3.3.]{smc:theory:Dm04}
\begin{equation}
\sup_{x\in\mathsf{X}}\mathbb{E}_{N}\left[\left|T_{p,n}^{N,2}(x)\right|^{r}\right]^{1/r}\leq2\left\Vert \varphi\right\Vert \frac{B_{r}}{\sqrt{N}}\left(\frac{\epsilon^{+}}{\epsilon^{-}}\right)^{5}\sum_{m=p}^{n}\tilde{\rho}^{m-p}\leq2\left\Vert \varphi\right\Vert \frac{B_{r}}{\sqrt{N}}\left(\frac{\epsilon^{+}}{\epsilon^{-}}\right)^{5}\frac{1}{1-\tilde{\rho}}.\label{eq:T_2_final-1}
\end{equation}
where the bound of Proposition \ref{prop_generic_stab} in equation
(\ref{eq:K_ratio_bounded}) has been applied to the left factor in
(\ref{eq:one_term-1}).

It remains to consider $T_{p,n}^{N,3}(x)$, and we do so using the
decomposition:
\begin{eqnarray}
\left|T_{p,n}^{N,3}(x)\right| & = & \left|\frac{\Phi\left(\eta_{p-1}^{N}\right)\left[\varphi\phi_{p-1}\left(x,\cdot\right)Q^{(n-p)}(1)\right]}{\Phi\left(\eta_{p-1}^{N}\right)Q^{\left(n-p\right)}(1)}-\frac{\eta_{p}\left[\varphi\phi_{p-1}\left(x,\cdot\right)Q^{(n-p)}(1)\right]}{\eta_{p}Q^{\left(n-p\right)}(1)}\right|\nonumber \\
 & \leq & \left\Vert \varphi\right\Vert \frac{\eta_{p-1}^{N}Q\left[\phi_{p-1}(x,\cdot)Q^{\left(n-p\right)}(1)\right]}{\eta_{p-1}^{N}Q^{\left(n-p+1\right)}(1)}\frac{\left|\left(\eta_{p-1}-\eta_{p-1}^{N}\right)Q^{\left(n-p+1\right)}(1)\right|}{\eta_{p-1}Q^{\left(n-p+1\right)}(1)}\nonumber \\
 &  & +\left\Vert \varphi\right\Vert \frac{\left|\left(\eta_{p-1}^{N}-\eta_{p-1}\right)Q\left[\phi_{p-1}(x,\cdot)Q^{\left(n-p\right)}(1)\right]\right|}{\eta_{p-1}Q^{\left(n-p+1\right)}(1)}.\label{eq:T3_decomp}
\end{eqnarray}
Now note that due to Lemma \ref{lem:Q_n_bounds} and the bound of
(\ref{eq:phi_bound}), 
\begin{eqnarray}
\sup_{x,y}\frac{Q\left[\phi_{p-1}(x,\cdot)Q^{\left(n-p\right)}(1)\right](y)}{\eta_{p-1}Q^{\left(n-p+1\right)}(1)} & \leq & \sup_{x,x^{\prime}}\left|\phi_{p-1}(x,x^{\prime})\right|\sup_{y}\frac{Q^{\left(n-p+1\right)}(1)(y)}{\eta_{p-1}Q^{\left(n-p+1\right)}(1)}\nonumber \\
 & \leq & \left(\frac{\epsilon^{+}}{\epsilon^{-}}\right)^{2},\label{eq:T3_prelim_bound}
\end{eqnarray}
and the same bound holds with $\eta_{p-1}^{N}$ in place of $\eta_{p-1}$.
Then Proposition \ref{prop_eta_uniform} combined with (\ref{eq:T3_prelim_bound})
may be applied to each of the terms in (\ref{eq:T3_decomp}) to yield:
\begin{equation}
\sup_{x\in\mathsf{X}}\mathbb{E}_{N}\left[\left|T_{p,n}^{N,3}(x)\right|^{r}\right]^{1/r}\leq\left\Vert \varphi\right\Vert \frac{B_{r}}{\sqrt{N}}4\left(\frac{\epsilon^{+}}{\epsilon^{-}}\right)^{7}.\label{eq:T_3_final-1}
\end{equation}
Combining (\ref{eq:T_1_final-1}), (\ref{eq:T_2_final-1}) and (\ref{eq:T_3_final-1})
completes the proof.\end{proof}
\begin{rem}
The treatment of the term $T_{p,n}^{N,2}$ in the proof uses some
arguments from \citep[Proof of Theorem 3.2]{del2010backward}, with
variations customized to the present context.\end{rem}
\begin{proof}
(of Theorem \ref{thm:L_r_bounds}) Consider the decomposition
\begin{eqnarray}
h_{p,n}^{N}(x)-h_{\star}(x) & = & \frac{Q_{p+1}^{N}(h_{p+1,n}^{N})(x)}{\lambda_{p}^{N}}-\frac{Q(h_{p+1,n})(x)}{\lambda_{p}}\nonumber \\
 &  & +h_{p,n}(x)-h_{\star}(x).\label{eq:h_thm_decomp}
\end{eqnarray}
The first difference on the r.h.s. of (\ref{eq:h_thm_decomp}) is
dealt with using Proposition \ref{prop:intermediate_L_r} applied
with $\varphi=1$. For the other difference, we have that by Proposition
\ref{prop:_h_n_h_star_bound}, 
\begin{equation}
\sup_{x\in\mathsf{X}}\left|h_{p,n}(x)-h_{\star}(x)\right|\leq C_{h}\rho^{(n-p)\wedge p}.\label{eq:T_4_final}
\end{equation}

To prove (\ref{eq:P_n_Lp}), consider the decomposition:
\[
P_{(p,n)}^{N}\left(x,A\right)-P_{\star}\left(x,A\right)=\Xi_{1}(x,A)+\Xi_{2}(x,A)+\Xi_{3}(x,A)
\]
where
\begin{eqnarray}
\Xi_{1}(x,A) & := & \frac{1}{h_{p-1,n}^{N}(x)}\left[\frac{Q_{p}^{N}(h_{p,n}^{N}\mathbb{I}_{A})(x)}{\lambda_{p-1}^{N}}-\frac{Q(h_{p,n}\mathbb{I}_{A})(x)}{\lambda_{p-1}}\right]\label{eq:P_decomp1}\\
\Xi_{2}(x,A) & := & \frac{Q(h_{p,n}\mathbb{I}_{A})(x)}{\lambda_{p-1}}\left[\frac{1}{h_{p-1,n}^{N}(x)}-\frac{1}{h_{p-1,n}(x)}\right]\label{eq:P_decomp2}\\
\Xi_{3}(x,A) & := & P_{(p,n)}(x,A)-P_{\star}(x,A).\label{eq:P_decomp3}
\end{eqnarray}
For the first term,
\begin{eqnarray*}
\mathbb{E}_{N}\left[\left|\Xi_{1}(x,A)\right|^{r}\right]^{1/r} & \leq & \frac{\epsilon^{+}}{\epsilon^{-}}\mathbb{E}\left[\left|\frac{Q_{p}^{N}(h_{p,n}^{N}\mathbb{I}_{A})(x)}{\lambda_{p-1}^{N}}-\frac{Q(h_{p,n}\mathbb{I}_{A})(x)}{\lambda_{p-1}}\right|^{r}\right]^{1/r}\\
 & \leq & 2\frac{\epsilon^{+}}{\epsilon^{-}}\frac{B_{r}}{\sqrt{N}}\tilde{C},
\end{eqnarray*}
where the first inequality uses the a lower bounds on $h_{p-1,n}^{N}(x)$
from Lemma \ref{lem:randomized_mixing} and the second inequality
is due to Proposition \ref{prop:intermediate_L_r} applied with $\varphi=\mathbb{I}_{A}$. 

We also have
\begin{eqnarray*}
\mathbb{E}_{N}\left[\left|\Xi_{2}(x,A)\right|^{r}\right]^{1/r} & \leq & \frac{\epsilon^{+}}{\epsilon^{-}}\frac{Q(h_{p,n}\mathbb{I}_{A})(x)}{\lambda_{p-1}h_{p-1,n}(x)}\mathbb{E}\left[\left|h_{p-1,n}(x)-h_{p-1,n}^{N}(x)\right|^{r}\right]^{1/r}\\
 & \leq & 2\frac{\epsilon^{+}}{\epsilon^{-}}\frac{B_{r}}{\sqrt{N}}\tilde{C},
\end{eqnarray*}
where for the first inequality the lower bound on $h_{p-1,n}^{N}(x)$
from Lemma \ref{lem:randomized_mixing} has been again be used, the
second inequality is due Lemma \ref{lem:approx_eigen} and Proposition
\ref{prop:intermediate_L_r} applied with $\varphi=1$. The term $\Xi_{3}$
is dealt with using Proposition \ref{prop:_h_n_h_star_bound} and
that completes the proof.
\end{proof}

\subsection{Proofs of Propositions \ref{prop:twsited_sampling} and \ref{prop:chi_square_bound}}
\begin{proof}
(of Proposition \ref{prop:twsited_sampling}) From (\ref{eq:twisted_sampling})
and the definition of $P_{(n+p,2n)}^{N}$ , for any $x_{0}\in\mathsf{X},$
\begin{eqnarray}
 &  & \mathbb{E}_{N}\left[\mathbb{E}_{N}\left[\left.F(X_{0:m})\frac{h_{n,2n}^{N}(X_{0})}{h_{n+m,2n}^{N}(X_{m})}\prod_{p=0}^{m-1}\frac{\lambda_{n+p}^{N}}{G_{\alpha}(X_{p})}\right|\mathcal{F}_{2n}\right]\right]\nonumber \\
 &  & =\mathbb{E}_{N}\left[\int_{\mathsf{X}^{m+1}}F(x_{0:m})\frac{h_{n,2n}^{N}(x_{0})}{h_{n+m,2n}^{N}(x_{m})}\prod_{p=1}^{m}\frac{\lambda_{n+p-1}^{N}}{G_{\alpha}(x_{p-1})}P_{(n+p,2n)}^{N}(x_{p-1,}dx_{p})\right]\nonumber \\
 &  & =\mathbb{E}_{N}\left[\int_{\mathsf{X}^{m+1}}F(x_{0:m})\prod_{p=1}^{m}\frac{1}{G_{\alpha}(x_{p-1})}Q_{n+p,2n}^{N}(x_{p-1,}dx_{p})\right]\nonumber \\
 &  & =\mathbb{E}_{N}\left[\int_{\mathsf{X}^{m+1}}F(x_{0:m})\prod_{p=1}^{m}\frac{\mathrm{d}M(x_{p-1},\cdot)}{\mathrm{d}\Phi(\eta_{n+p-1}^{N})}(x_{p})\eta_{n+p}^{N}(dx_{p})\right].\label{eq:sample_twist_conditional}
\end{eqnarray}
where $\mathcal{F}_{2n}$ is the $\sigma$-algebra generated by the
particle system at time $2n$. We will proceed to decompose the difference
between (\ref{eq:sample_twist_conditional}) and $\pi_{m}\left(\delta\right)$. 

For $\ell=1,...,m,$ define $F_{\ell}$ by
\[
F_{m}(x_{0:m}):=F(x_{0:m}),\quad\quad F_{\ell}(x_{0:\ell}):=\int_{\mathsf{X}}F_{\ell+1}(x_{0:\ell+1})M(x_{\ell},dx_{\ell+1}),\quad\ell=1,...,m-1,
\]
and observe that then
\begin{equation}
M(F_{1})(x)=\mathbb{E}_{x}\left[F(X_{0:m})\right].\label{eq:F_0}
\end{equation}
For any $\ell=0,...,m$, and $x_{0}\in\mathsf{X}$ , define
\begin{equation}
\overline{F}_{0}^{N}(x_{0}):=M(F_{1})(x_{0}),\quad\quad\overline{F}_{\ell}^{N}(x_{0}):=\int_{\mathsf{X}^{\ell}}F_{\ell}\left(x_{1:\ell}\right)\prod_{p=1}^{\ell}\frac{\mathrm{d}M(x_{p-1},\cdot)}{\mathrm{d}\Phi(\eta_{n+p-1}^{N})}(x_{p})\eta_{n+p}^{N}(dx_{p}),\quad\ell=1,...,m.\label{eq:F-bar_defn}
\end{equation}
Then for any $\ell=2,...,m$, 
\begin{eqnarray}
 &  & \mathbb{E}_{N}\left[\left.\overline{F}_{\ell}^{N}(x_{0})\right|\mathcal{F}_{n+\ell-1}\right]\nonumber \\
 &  & =\int_{\mathsf{X}^{\ell-1}}\prod_{p=1}^{\ell-1}\frac{\mathrm{d}M(x_{p-1},\cdot)}{\mathrm{d}\Phi(\eta_{n+p-1}^{N})}(x_{p})\eta_{n+p}^{N}(dx_{p})\mathbb{E}_{N}\left[\left.\int_{\mathsf{X}}F_{\ell}(x_{1:\ell})\frac{dM(x_{\ell-1},\cdot)}{d\Phi(\eta_{n+\ell-1}^{N})}(x_{\ell})\eta_{n+\ell}^{N}(dx_{\ell})\right|\mathcal{F}_{n+\ell-1}\right]\nonumber \\
 &  & =\int_{\mathsf{X}^{\ell-1}}\prod_{p=1}^{\ell-1}\frac{\mathrm{d}M(x_{p-1},\cdot)}{\mathrm{d}\Phi(\eta_{n+p-1}^{N})}(x_{p})\eta_{n+p}^{N}(dx_{p})\int_{\mathsf{X}}F_{\ell}(x_{1:\ell})M(x_{\ell-1},dx_{\ell})\nonumber \\
 &  & =\int_{\mathsf{X}^{\ell-1}}F_{\ell-1}(x_{1:\ell-1})\prod_{p=1}^{\ell-1}\frac{dM(x_{p-1},\cdot)}{d\Phi(\eta_{n+p-1}^{N})}(x_{p})\eta_{n+p}^{N}(dx_{p})=\overline{F}_{\ell-1}^{N}(x_{0}),\label{eq:one_step}
\end{eqnarray}
and a similar manipulation shows 
\begin{equation}
\mathbb{E}_{N}\left[\left.\overline{F}_{1}^{N}(x_{0})\right|\mathcal{F}_{n}\right]=\overline{F}_{0}^{N}(x_{0}).\label{eq:F_bar_first_step.}
\end{equation}
We then have that 
\[
\mathbb{E}_{N}\left[\overline{F}_{m}^{N}(x_{0})\right]-\mathbb{E}_{x_{0}}\left[F(X_{0:m})\right]=\sum_{\ell=1}^{m}\mathbb{E}_{N}\left[\overline{F}_{\ell}^{N}(x_{0})-\overline{F}_{\ell-1}^{N}(x_{0})\right]=0,
\]
where (\ref{eq:F_0}), (\ref{eq:one_step}), (\ref{eq:F_bar_first_step.})
and (\ref{eq:F-bar_defn}) have been applied. But $\overline{F}_{m}^{N}(x_{0})$
is just what appears inside the expectation (\ref{eq:sample_twist_conditional}),
so the proof is complete.\end{proof}
\begin{lem}
\label{lem:prod_lambda_bound} Assume \textbf{(H)} and let $\mathbb{E}_{N}$
denote the expectation w.r.t. the joint law of the particle system
and $(X_{p})$ sampled according to (\ref{eq:twisted_sampling}).
There exists a finite constant $C$ such that for all $m\geq1$, $N\geq1$,
\[
\sup_{n\geq0}\mathbb{E}_{N}\left[\left(\prod_{p=0}^{m-1}\frac{\lambda_{n+p}^{N}}{\lambda_{n+p}}-1\right)^{2}\right]\leq\left(1+\frac{C}{\sqrt{N}}\right)\left[\left(1+\frac{C}{N}\right)^{m}-1\right]
\]
\end{lem}
\begin{proof}
Throughout the proof $C$ denotes a finite constant which is independent
of $m$, $n$ and $N$, but whose value may change on each appearance.
From hereon $m\geq1$, $N\geq1$ and $n\geq0$ are fixed to arbitrary
values.

For $1\leq p\leq m$, consider the decomposition
\[
\prod_{q=0}^{p-1}\frac{\lambda_{n+q}^{N}}{\lambda_{n+q}}-1=\sum_{q=0}^{p}\Delta_{p,q}
\]
where
\begin{eqnarray*}
\Delta_{p,0} & := & \left[\eta_{n}^{N}-\eta_{n}\right]\frac{Q^{(p)}(1)}{\eta_{n}Q^{(p)}(1)}\\
\Delta_{p,q} & := & \left(\prod_{r=0}^{q-1}\frac{\lambda_{n+r}^{N}}{\lambda_{n+r}}\right)\left[\eta_{n+q}^{N}-\frac{\eta_{n+q-1}^{N}Q}{\lambda_{n+q-1}^{N}}\right]\frac{Q^{(p-q)}(1)}{\eta_{n+q}Q^{(p-q)}(1)},\quad1\leq q\leq p.
\end{eqnarray*}
Note that by Lemma \ref{lem:Q_n_bounds}, $\sup_{p}\sup_{x}Q^{(p)}(1)(x)/\eta_{n}Q^{(p)}(1)\leq\epsilon^{+}/\epsilon^{-}$,
so by Proposition \ref{prop_eta_uniform},

\[
\sup_{p}\left|\mathbb{E}_{N}\left[\Delta_{p,0}\right]\right|\leq\frac{C}{\sqrt{N}},\quad\quad\sup_{p}\mathbb{E}_{N}\left[\left|\Delta_{p,0}\right|^{2}\right]\leq\frac{C}{N}.
\]
Also note that 
\[
\eta_{n+q}^{N}-\frac{\eta_{n+q-1}^{N}Q}{\lambda_{n+q-1}^{N}}=\eta_{n+q}^{N}-\Phi(\eta_{n+q-1}^{N})
\]
and recall that given $\mathcal{F}_{n+q-1}$, $(\zeta_{n+q}^{i})_{i=1}^{N}$
are conditionally i.i.d. draws from $\Phi(\eta_{n+p-1}^{N})$. Therefore
\[
\mathbb{E}\left[\left.\Delta_{p,q}\right|\mathcal{F}_{2n}\right]=0\quad\text{ and}\quad\mathbb{E}\left[\left.\Delta_{p,q}\Delta_{p,l}\right|\mathcal{F}_{2n}\right]=0,\quad1\leq q<l\leq p,
\]
so
\[
\mathbb{E}_{N}\left[\prod_{q=0}^{p-1}\frac{\lambda_{n+q}^{N}}{\lambda_{n+q}}-1\right]=\mathbb{E}_{N}\left[\Delta_{p,0}\right].
\]
Collecting the above and adopting the convention $\prod_{r=0}^{-1}\frac{\lambda_{n+r}^{N}}{\lambda_{n+r}}=1$,
we have
\begin{eqnarray*}
\mathbb{E}_{N}\left[\left(\prod_{q=0}^{p-1}\frac{\lambda_{n+q}^{N}}{\lambda_{n+q}}-1\right)^{2}\right] & = & \sum_{q=0}^{p}\mathbb{E}_{N}\left[(\Delta_{p,q})^{2}\right]\\
 & \leq & \frac{C}{N}\sum_{q=0}^{p-1}\mathbb{E}_{N}\left[\left(\prod_{r=0}^{q-1}\frac{\lambda_{n+r}^{N}}{\lambda_{n+r}}\right)^{2}\right]\\
 & = & \frac{C}{N}\sum_{q=0}^{p-1}\mathbb{E}_{N}\left[\left(\prod_{r=0}^{q-1}\frac{\lambda_{n+r}^{N}}{\lambda_{n+r}}-1+1\right)^{2}\right]\\
 & \leq & \frac{C}{N}\sum_{q=0}^{p-1}\left(\mathbb{E}_{N}\left[\left(\prod_{r=0}^{q-1}\frac{\lambda_{n+r}^{N}}{\lambda_{n+r}}-1\right)^{2}\right]+1+2\left|\mathbb{E}\left[\Delta_{q,0}\right]\right|\right)\\
 & \leq & \frac{C}{N}\sum_{q=0}^{p-1}\left(\mathbb{E}_{N}\left[\left(\prod_{r=0}^{q-1}\frac{\lambda_{n+r}^{N}}{\lambda_{n+r}}-1\right)^{2}\right]+1+\frac{C}{\sqrt{N}}\right).
\end{eqnarray*}
With the shorthand
\[
a_{p}:=\mathbb{E}_{N}\left[\left(\prod_{q=0}^{p-1}\frac{\lambda_{n+q}^{N}}{\lambda_{n+q}}-1\right)^{2}\right],\quad0\leq p\leq m,\quad b:=1+\frac{C}{\sqrt{N}},
\]
we have so far established
\begin{equation}
a_{0}=0,\quad\quad a_{p}\leq\frac{C}{N}\sum_{q=0}^{p-1}(a_{q}+b),\quad1\leq p\leq m.\label{eq:a_p_recursion}
\end{equation}
We claim that solving this recursion gives
\begin{equation}
a_{p}\leq b\left[\left(1+\frac{C}{N}\right)^{p}-1\right].\label{eq:a_p_sol}
\end{equation}
Indeed (\ref{eq:a_p_sol}) holds with $p=0$ since $a_{0}=0$ by definition,
and when (\ref{eq:a_p_sol}) holds at ranks less than or equal to
$p$, (\ref{eq:a_p_recursion}) gives 
\begin{eqnarray*}
a_{p+1} & \leq & \frac{C}{N}\sum_{q=0}^{p}\left(b\left[\left(1+\frac{C}{N}\right)^{q}-1\right]+b\right)\\
 & = & b\frac{C}{N}\frac{\left(1+\frac{C}{N}\right)^{p+1}-1}{\left(1+\frac{C}{N}\right)-1}=b\left[\left(1+\frac{C}{N}\right)^{p+1}-1\right].
\end{eqnarray*}
The proof is complete since (\ref{eq:a_p_sol}) with $p=m$ is the
bound in the statement of the lemma.\end{proof}
\begin{lem}
\label{lem:h_bounds}Assume the assumptions of Lemma \ref{lem:prod_lambda_bound}
hold and in addition that $\mathsf{X}$ is a finite set. There exists
a finite constant $C$ such that for all $1\leq m\leq n$ and $N\geq1$,
\begin{eqnarray*}
\left|\prod_{p=0}^{m-1}\frac{\lambda_{n+p}}{\lambda_{\star}}-1\right| & \leq & \left(1-\frac{\epsilon^{-}}{\epsilon^{+}}\right)^{n}C\\
\mathbb{E}_{N}\left[\left(\frac{h_{\star}(X_{m})}{h_{n+m,2n}^{N}(X_{m})}-1\right)^{2}\right]^{1/2} & \leq & C\left[\frac{1}{\sqrt{N}}+\left(1-\frac{\epsilon^{-}}{\epsilon^{+}}\right)^{n-m}\right]\mathrm{card}(\mathsf{X})\\
\mathbb{E}_{N}\left[\left(\frac{h_{n,2n}^{N}(X_{0})}{h_{\star}(X_{0})}-1\right)^{2}\right]^{1/2} & \leq & C\left[\frac{1}{\sqrt{N}}+\left(1-\frac{\epsilon^{-}}{\epsilon^{+}}\right)^{n}\right]
\end{eqnarray*}
\end{lem}
\begin{proof}
By Proposition \ref{prop:_h_n_h_star_bound},
\begin{eqnarray*}
\left|\prod_{p=0}^{m-1}\frac{\lambda_{n+p}}{\lambda_{\star}}-1\right| & = & \left|\frac{\eta_{n}Q^{(m)}(1)}{\eta_{\star}Q^{(m)}(1)}-1\right|=\left|\left[\eta_{n}-\eta_{\star}\right]\frac{Q^{(m)}(1)}{\eta_{\star}Q^{(m)}(1)}\right|\leq\left(1-\frac{\epsilon^{-}}{\epsilon^{+}}\right)^{n}C_{\eta}\frac{\epsilon^{+}}{\epsilon^{-}}.
\end{eqnarray*}

For the second inequality in the statement, using Lemma \ref{lem:randomized_mixing}
and noting that by assumption $\mathsf{X}$ is a finite set, we have
\begin{equation}
\left|\frac{h_{\star}(X_{m})}{h_{n+m,2n}^{N}(X_{m})}-1\right|\leq\max_{x\in\mathsf{X}}\left|\frac{h_{\star}(x)}{h_{n+m,2n}^{N}(x)}-1\right|\leq\frac{\epsilon^{+}}{\epsilon^{-}}\sum_{x\in\mathsf{X}}\left|h_{\star}(x)-h_{n+m,2n}^{N}(x)\right|.\label{eq:pathwise_h_diff}
\end{equation}
Theorem \ref{thm:L_r_bounds} together with Minkowski's inequality
applied to (\ref{eq:pathwise_h_diff}) gives the desired bound. The
third inequality is proved similarly, except that under (\ref{eq:twisted_sampling})
$X_{0}=x$ a.s., hence
\[
\left|\frac{h_{\star}(X_{m})}{h_{n+m,2n}^{N}(X_{m})}-1\right|=\left|\frac{h_{\star}(x)}{h_{n+m,2n}^{N}(x)}-1\right|,\quad a.s.
\]

\end{proof}

\begin{proof}
(of Proposition \ref{prop:chi_square_bound}) Throughout the proof
$m$, $N$ and $n$ are fixed. Define

\[
W:=\frac{h_{n,2n}^{N}(X_{0})}{h_{\star}(X_{0})}\frac{h_{\star}(X_{m})}{h_{n+m,2n}^{N}(X_{m})}\prod_{p=0}^{m-1}\frac{\lambda_{n+p}^{N}}{\lambda_{\star}},
\]
so that
\[
\frac{\mathrm{d}\overline{\mathbb{P}}_{x}}{\mathrm{d}\overline{\mathbb{P}}_{x}^{N,n}}(X_{0},\ldots,X_{m})=\mathbb{E}_{N}\left[\left.W\right|X_{0},\ldots,X_{m}\right].
\]
For the result of the Proposition we need to bound $\mathbb{E}_{N}\left[\mathbb{E}_{N}\left[\left.W-1\right|X_{0},\ldots,X_{m}\right]^{2}\right]$
by the r.h.s. of (\ref{eq:chi_square_bound}). By the conditional
Jensen's inequality, it is sufficient to show that the same upper
bound holds for $\mathbb{E}_{N}\left[(W-1)^{2}\right]$.

Consider the decomposition $W-1=\sum_{i=1}^{4}W_{i}$ where 
\begin{eqnarray*}
W_{1} & := & \frac{h_{n,2n}^{N}(X_{0})}{h_{\star}(X_{0})}\frac{h_{\star}(X_{m})}{h_{n+m,2n}^{N}(X_{m})}\left(\prod_{p=0}^{m-1}\frac{\lambda_{n+p}}{\lambda_{\star}}\right)\left(\prod_{p=0}^{m-1}\frac{\lambda_{n+p}^{N}}{\lambda_{n+p}}-1\right),\\
W_{2} & := & \frac{h_{n,2n}^{N}(X_{0})}{h_{\star}(X_{0})}\frac{h_{\star}(X_{m})}{h_{n+m,2n}^{N}(X_{m})}\left(\prod_{p=0}^{m-1}\frac{\lambda_{n+p}}{\lambda_{\star}}-1\right),\\
W_{3} & := & \frac{h_{n,2n}^{N}(X_{0})}{h_{\star}(X_{0})}\left(\frac{h_{\star}(X_{m})}{h_{n+m,2n}^{N}(X_{m})}-1\right),\\
W_{4} & := & \frac{h_{n,2n}^{N}(X_{0})}{h_{\star}(X_{0})}-1.
\end{eqnarray*}
By (\ref{eq:eta_in_P_h_bound}) and Lemma \ref{lem:randomized_mixing}
\begin{equation}
\sup_{x}\frac{h_{n,2n}^{N}(x)}{h_{\star}(x)}\vee\frac{h_{\star}(x)}{h_{n+m,2n}^{N}(x_{m})}\leq\left(\frac{\epsilon^{+}}{\epsilon^{-}}\right)^{2}.\label{eq:H_h_N_uniform_bounds}
\end{equation}
 Since
\[
\prod_{p=0}^{m-1}\frac{\lambda_{n+p}}{\lambda_{\star}}=\frac{\eta_{n}Q^{(m)}(1)}{\eta_{\star}Q^{(m)}(1)}\leq\frac{\epsilon^{+}}{\epsilon^{-}},
\]
Lemma \ref{lem:prod_lambda_bound} gives
\[
\mathbb{E}_{N}\left[(W_{1})^{2}\right]^{1/2}\leq\left(\frac{\epsilon^{+}}{\epsilon^{-}}\right)^{5}\mathbb{E}_{N}\left[\left(\prod_{p=0}^{m-1}\frac{\lambda_{n+p}^{N}}{\lambda_{n+p}}-1\right)^{2}\right]^{1/2}\leq C\left(1+\frac{C}{\sqrt{N}}\right)^{1/2}\left[\left(1+\frac{C}{N}\right)^{m}-1\right]^{1/2}.
\]
Lemma \ref{lem:h_bounds} and (\ref{eq:H_h_N_uniform_bounds}) give
\begin{eqnarray*}
\mathbb{E}_{N}\left[(W_{2})^{2}\right]^{1/2} & \leq & \left(\frac{\epsilon^{+}}{\epsilon^{-}}\right)^{4}\left(\prod_{p=0}^{m-1}\frac{\lambda_{n+p}}{\lambda_{\star}}-1\right)\leq C\left(1-\frac{\epsilon^{-}}{\epsilon^{+}}\right)^{n},
\end{eqnarray*}

\begin{eqnarray*}
\mathbb{E}_{N}\left[(W_{3})^{2}\right]^{1/2} & \leq & \left(\frac{\epsilon^{+}}{\epsilon^{-}}\right)^{2}\mathbb{E}_{N}\left[\left(\frac{h_{\star}(X_{m})}{h_{n+m,2n}^{N}(X_{m})}-1\right)^{2}\right]^{1/2}\leq C\left[\frac{1}{\sqrt{N}}+\left(1-\frac{\epsilon^{-}}{\epsilon^{+}}\right)^{n-m}\right]\mathrm{card}(\mathsf{X}),
\end{eqnarray*}

\[
\mathbb{E}_{N}\left[(W_{4})^{2}\right]^{1/2}\leq C\left[\frac{1}{\sqrt{N}}+\left(1-\frac{\epsilon^{-}}{\epsilon^{+}}\right)^{n}\right].
\]
Combining these bounds with Minkowski's inequality applied to $W-1=\sum_{i=1}^{4}W_{i}$
completes the proof of the proposition.
\end{proof}
\bibliographystyle{plainnat}
\bibliography{SMC_stability}

\end{document}